\documentclass[reqno]{amsart}
\usepackage[a4paper]{geometry}
\usepackage{inconsolata}
\usepackage{bm}
\usepackage{stmaryrd}
\usepackage{graphicx}
\usepackage{caption}
\usepackage{subcaption}
\usepackage{rotating}
\usepackage[normalem]{ulem}

\usepackage[final,kerning=true,spacing=true,factor=1100,stretch=10,shrink=10]{microtype}

\newtheorem{theorem}{\bf Theorem}[section]

\usepackage{algorithm}
\usepackage{algpseudocode}
\makeatletter
\def\BState{\State\hskip-\ALG@thistlm}
\makeatother

\newcommand{\abs}[1]{{\lvert#1\rvert}} 
\newcommand{\norm}[1]{{\lVert#1\rVert}} 
\newcommand{\jump}[1]{{\left\llbracket#1\right\rrbracket}}
\newcommand{\elprod}{\circ} 
\newcommand{\jacobian}[1]{\vec{\nabla}#1} 
\renewcommand{\d}{\operatorname{d}} 
\newcommand{\dx}{\d\!\vx} 
\newcommand{\dt}{\d\!t} 
\newcommand{\ds}{\d\!s} 
\newcommand{\ddt}[1]{{#1}_{t}} 
\newcommand{\ddn}[1]{\frac{\partial #1}{\partial \vec{n}}} 
\renewcommand{\vec}[1]{\bm{#1}} 
\newcommand{\vx}{\vec{x}} 
\newcommand{\diag}[1]{\widehat{#1}} 
\renewcommand{\Re}{\mathbb{R}} 

\renewcommand{\u}{u} 
\newcommand{\uvec}[1][]{{\vec{u}_{#1}}} 
\newcommand{\uhvec}{\uvec_h} 
\newcommand{\uhtvec}{\uvec_{h,t}} 
\newcommand{\uh}{\u_h} 
\newcommand{\uht}{\u_{h,t}} 
\newcommand{\vh}{v_h}
\newcommand{\vvec}{\vec{v}}

\newcommand{\wvec}{\vec{w}}

\newcommand{\initialcond}{\vec{\psi}} 
\newcommand{\diff}[1][]{\varepsilon_{#1}} 
\newcommand{\growth}[1][]{a_{#1}} 
\newcommand{\reac}[2]{\alpha_{#1,#2}} 
\newcommand{\f}{f} 
\newcommand{\spacedim}{d}
\newcommand{\domain}{\Omega}
\newcommand{\boundary}{\partial\domain}

\newcommand{\Vh}{V_h} 
\newcommand{\Th}{\mathcal{T}_h} 
\newcommand{\edges}{\mathcal{E}_h} 
\newcommand{\element}{T} 
\newcommand{\edge}{e} 
\renewcommand{\k}{k} 
\newcommand{\testfn}{\chi} 
\newcommand{\vectestfn}{\vec{\testfn}}
\newcommand{\elliprecon}[1]{\mathcal{R}#1} 
\newcommand{\error}{e} 
\newcommand{\reg}{r} 
\newcommand{\timestep}{n} 
\newcommand{\fh}{\f_h} 

\newcommand{\expconst}{(1 + 4M \Cf)}
\newcommand{\Cdom}{C_{\domain}}
\newcommand{\Cf}{C_{\f}}
\newcommand{\Celip}{C_{\text{elip}}}

\begin{document}

\title[New dynamical patterns in a cyclic competition model]{Revealing new dynamical patterns\\ in a reaction-diffusion model with cyclic competition\\ via a novel computational framework}

\author{A. Cangiani}
\author{E. H. Georgoulis}
\author{A. Yu. Morozov}
\author{O. J. Sutton}

\address[A. Cangiani and A. Yu. Morozov]{Department of Mathematics, University of Leicester, University Road, Leicester LE1 7RH, UK}
\email{Andrea.Cangiani@le.ac.uk, am379@le.ac.uk}
\address[E. H. Georgoulis]{Department of Mathematics, University of Leicester, University Road, Leicester LE1 7RH, UK \emph{and} Department of Mathematics, School of Mathematical and Physical Sciences, National Technical University of Athens, Zografou 157 80, Greece} \email{Emmanuil.Georgoulis@le.ac.uk}
\address[O. J. Sutton]{Department of Mathematics and Statistics, University of Reading,
Whiteknights, PO Box 220, Reading RG6 6AX, UK}
\email{O.Sutton@reading.ac.uk}


\begin{abstract}
Understanding how patterns and travelling waves form in chemical and biological reaction-diffusion models is an area which has been widely researched, yet is still experiencing fast development.
Surprisingly enough, we still do not have a clear understanding about all possible types of dynamical regimes in classical reaction-diffusion models such as Lotka-Volterra competition models with spatial dependence.
In this work, we demonstrate some new types of wave propagation and pattern formation in a classical three species cyclic competition model with spatial diffusion, which have been so far missed in the literature.
These new patterns are characterised by a high regularity in space, but are different from patterns previously known to exist in reaction-diffusion models, and may have important applications in improving our understanding of biological pattern formation and invasion theory.
Finding these new patterns is made technically possible by using an automatic adaptive finite element method driven by a novel a posteriori error estimate which is proven to provide a reliable bound for the error of the numerical method.
We demonstrate how this numerical framework allows us to easily explore the dynamical patterns both in two and three spatial dimensions.
\end{abstract}

\maketitle

\section{Introduction}
\label{sec:introduction}

The formation of patterns and travelling waves in chemical, physical, biological and game theoretical models described by systems of advection-reaction-diffusion equations are widely researched phenomena discussed in an immense number of publications.
The pioneering works in this area were published as early as the late 1930s with the works~\cite{Fisher37,kolmogorov37} on wave propagation and Turing's demonstration of pattern forming mechanisms in the early 1950s~\cite{turing1952}.
A nice introduction to the current state of the art can be found in~\cite{Cooper12,satnoianu2000, volpert2009}.
To this day, the research area continues to experience a fast development: new types of patterns and waves have recently been reported in more complicated systems as compared to the initial models.
For example, a number of exotic patterns, such as envelope and multi-envelope quasi-solitons~\cite{tsyganov2014}, have been observed in systems with cross-diffusion (i.e. where the diffusivity matrix is non-diagonal).
Similarly, novel patterns have also been found in reaction-diffusion systems with non-local terms, i.e. in the case where the reaction terms depend not only on the local species densities but also on the species distribution in some neighbourhood, described by an integral term~\cite{A15, volpert2009}.
Furthermore, models of interactions on growing domains or in the case when the spatial parameters are variable can result in the emergence of new types of patterns with different properties from those on fixed domains or with homogeneous parameters~\cite{CGM02, page2005}.

Interestingly, despite the current trend towards investigating the behaviour of more and more sophisticated reaction-diffusion systems (e.g. including density-dependent diffusion coefficients, cross-diffusion, time delay, nonlocal terms, etc), we still miss a fundamental understanding of the possible types of dynamical regimes in well known reaction-diffusion models, some of which are included in textbooks on mathematical biology.

A notable example is the Lotka-Volterra competition model in space.
This model is well known in the literature with numerous applications in mathematical biology and areas of game theory such as voting models~\cite{Petrovskii2001, mimura2015, MuN11, AM12, Ch12, C15}.
In particular, this system is often considered as a paradigm for biodiversity modelling.
New types of patterns have been recently demonstrated in this system which occur for spatially homogeneous diffusion coefficients~\cite{AM12, mimura2015, C15}.
This includes, for example, patchy invasion (the spread of a species via the formation and propagation of chaotic patches without a smooth population front), which was originally believed to occur only in predator-prey or inhibitor-activator types  models~\cite{petrovskii2002}.
Surprisingly enough, the Lotka-Volterra system still remains poorly understood, especially when the interacting species diffuse at different rates.
Here, we show the existence of several new dynamical patterns, related to the spread of travelling waves, which have been missed in the literature so far, and which may have important biological applications.
In particular, we demonstrate spreading patterns exhibiting complex regular spatial structure which have not been observed so far in reaction-diffusion models with a non-transitive competition such as the Lotka-Volterra cyclic model.

Note that another important gap in our knowledge about patterns and waves in reaction-diffusion models is that most existing results have been obtained for either one or two dimensions in space. The simple reason for this is that exploring interactions in three dimensions presents a challenge both from the analytical and the numerical point of view. This is rather a nuisance since many applications of these models, modelling microbiological interactions in water bodies, various medical applications as well as chemical interactions within a substantial volume, are inherently three dimensional.
In particular, an important question is how the dimension of the system's spatial domain influences species persistence and biodiversity~\cite{wilson1995, morozov2007}.

A thorough investigation of the Lotka-Volterra model with complex spatio-temporal patterns in two and three space dimensions is made possible here by applying a novel \emph{adaptive} numerical method, based on a Finite Element Method (FEM) coupled with a reliable \emph{a posteriori} error estimator. Our choice of method is motivated by the observation that solutions to this model (and many others) feature large patches of relative homogeneity in which a single species dominates, separated by narrow travelling fronts where interactions occur. These narrow fronts can only be resolved by a suitably fine computational mesh, although using such a fine mesh across the large domains and long time scales required for the full system dynamics to develop can be prohibitively expensive, a difficulty which is amplified as the number of spatial dimensions increases. Since this resolution is not required in large areas where a single species dominates, the Lotka-Volterra system is a perfect target for mesh adaptivity, tightly focussing the computational effort in the areas of interaction and sparsely deploying it elsewhere.

The algorithm we employ is based on mathematically rigorous error estimates which are computable at each time step because they only depend on `known' quantities such as the numerical solution and problem data.
A posteriori error estimation for both stationary and dynamic linear equations is now relatively well understood~\cite{AO00,V13,MN03}.
This approach is extended here by developing energy norm a posteriori error bounds for semilinear systems of parabolic equations for which the nonlinear reaction terms satisfy suitable growth conditions, ensuring that the framework we present also includes many other similar reaction-diffusion type systems typically encountered as models of biological and ecological phenomena. For this class of equations, discretised by finite element methods, we prove new a posteriori upper bounds of the true error of the numerical method. The analysis is based on the elliptic reconstruction technique of~\cite{MN03} and exploits, in an a posteriori fashion, an argument from~\cite{CGJ13} to bound the nonlinear terms; see also~\cite{CGKM16} for a similar argument applied to a single reaction-diffusion equation whose solution blows up in finite-time.
Based on the a posteriori error estimator, we devise a local error indicator which is used to drive the computational mesh adaptation algorithm used in all of the numerical simulations presented here.
The standard second order finite element method is judiciously combined with the second order Crank-Nicolson Adams-Bashforth implicit-explicit (IMEX) discretisation in time, allowing us to treat the linear part of the differential operator implicitly while the nonlinear reaction term is treated explicitly.
The resulting scheme therefore allows us to efficiently explore the spatio-temporal patterns arising in the Lotka-Volterra cyclic competition model in both two and three spatial dimensions. We demonstrate its effectiveness by revealing  several novel spatial-temporal patterns in the  model in two and three spatial dimensions, and explore their properties.

The paper is organised as follows. Section~\ref{sec:model} introduces the cyclic competition model with diffusion, while Section~\ref{sec:numericalscheme} provides the details of the novel numerical method used in our simulations. Some results from our simulations are presented in Section~\ref{sec:results}, demonstrating the new  dynamical patterns in the the cases of two and three spatial dimensions.
In Section~\ref{sec:apost} we present the proof of a reliable a posteriori error bound for a general class of semilinear parabolic problems, and discuss how it is used to drive mesh adaptivity.
A concluding discussion of the significance of our results from the biological and computational viewpoints is given in Section~\ref{sec:conclusion}.

\section{Cyclic competition in a reaction-diffusion model}
\label{sec:model}

The spatio-temporal interactions of three competing species are described using a reaction-diffusion scheme based on the three-species Lotka-Volterra competition  model, as proposed by May and Leonard~\cite{ML75} and~\cite{mimura2015}. The unscaled system in $\domain \times [0,T]$, with $\domain \subset \Re^{\spacedim}$ where $\spacedim = 2$ or $3$, reads
\begin{align*}
    \ddt{(\uvec[1])} &= D_1 \Delta \uvec[1] + \growth[1] \uvec[1] (1 - \reac{1}{1} \uvec[1] - \reac{1}{2} \uvec[2] - \reac{1}{3} \uvec[3]) \\
    \ddt{(\uvec[2])} &= D_2 \Delta \uvec[2] + \growth[2] \uvec[2] (1 - \reac{2}{1} \uvec[1] - \reac{2}{2} \uvec[2] - \reac{2}{3} \uvec[3]) \\
    \ddt{(\uvec[3])} &= D_3 \Delta \uvec[3] + \growth[3] \uvec[3] (1 - \reac{3}{1} \uvec[1] - \reac{3}{2} \uvec[2] - \reac{3}{3} \uvec[3]),
\end{align*}
where $\uvec[i]\equiv\uvec[i](t,\bold{x})$ is the density of species $i$ at time $t$ and location $\bold{x}$;  $a_i$ is the intrinsic growth rate of species $\uvec[i]$, and each coefficient $\alpha_{i,j}$ represents the limiting effect that the presence of species $\uvec[j]$ has on species $\uvec[i]$ (the term $\alpha_{i,i}$ describes self-limitation of the population). The diffusion coefficients $D_i$ describe the dispersal rate/mobility of each species.

The model is completed by homogeneous Neumann boundary conditions on $\boundary$ and the initial conditions
\begin{align*}
    \uvec[i](\vx, 0) &= \initialcond_i(\vx) \quad \text{ in } \domain,
\end{align*}
for various different choices of functions $\initialcond_i$, $i=1,2,3$.

Since the model contains a large number of parameters, it is convenient to scale each species'
density by introducing $\tilde{\uvec}_i = \alpha_{i,i}\uvec[i]$, and then rescale both time and space to obtain the simplified three-species competition system:
\begin{align}\label{eq:scaledsystem}
    \begin{split}
        \ddt{(\uvec[1])} &= \phantom{\diff[1]} \Delta \uvec[1] + \phantom{\growth[1]} \uvec[1] (1 - \uvec[1] - \reac{1}{2} \uvec[2] - \reac{1}{3} \uvec[3]) \\
        \ddt{(\uvec[2])} &= \diff[2] \Delta \uvec[2] + \growth[2] \uvec[2] (1 - \reac{2}{1} \uvec[1] - \uvec[2] - \reac{2}{3} \uvec[3]) \\
        \ddt{(\uvec[3])} &= \diff[3] \Delta \uvec[3] + \growth[3] \uvec[3] (1 - \reac{3}{1} \uvec[1] - \reac{3}{2} \uvec[2] - \uvec[3]).
    \end{split}
\end{align}
Note that we can assume without loss of generality that $\epsilon_2,\epsilon_3 \le 1$ since we can always choose to scale the system to the largest diffusion coefficient, and consider the corresponding species as the first one.

The Lotka-Volterra competition  model with diffusion has been studied in a number of papers where it was shown that it can possess complex patterns of dynamics~\cite{K82, MF86, Petrovskii2001, mimura2015, MuN11, AM12, Ch12, C15}.
The outcome of interactions between the species is strongly affected by the hierarchy of the competition structure, as dictated by examining each pairwise interaction.
One of the most interesting scenarios includes a cyclic competition structure, in which (roughly speaking) species 1 outcompetes species 2, species 2 outcompetes species 3 and species 3 outcompetes species 1.
Such cyclic dominance is analogous to the popular game of `rock-paper-scissors'.
Some well-known examples of cyclic interactions observed in nature include competition between side-blotched lizards~\cite{sinervo1996}, coral reef invertebrates~\cite{Buss79}, yeast strains~\cite{paquin1983}, and various bacterial strains~\cite{kirkup04}.
The same model also arises in non-biological situations such as many-player prisoner's dilemma games~\cite{HDHS02} or some types of voter models~\cite{tainaka1993}.

Formalising the above characterisation of cyclic competition can be tricky, although here we will follow the definition given by~\cite{AM12}.
This is based on considering the outcomes of pairwise interactions in an unbounded one-dimensional spatial domain (i.e. in the absence of a third species), starting from initial conditions such that the species densities at positive and negative infinity are equal to the carrying capacities for one species and zero for the other.
Cyclic competition is then said to occur if the direction of the resulting travelling waves preserves the cyclic order $1 > 2 > 3 > 1$.
For example, the domain occupied by species 2 at its carrying capacity level should eventually be replaced by a spreading wave of species 1.
This generic definition of cyclic dominance allows two main types of local dynamics~\cite{AM12, C15}.
In \emph{classical cyclic competition}, the phase portrait of each pairwise interaction should involve only one stable steady state corresponding to the presence of the stronger competitor at its carrying capacity.
In this case, adding a spatial dimension to the local interaction does not reverse the outcome of the competition since the corresponding travelling wave will be directed from the domain occupied by the stronger competitor to that of the weaker competitor~\cite{Hosono1998}.
The mathematical conditions for this to occur are: $\alpha_{i,i+1} \le 1 $ and $\alpha_{i+1,i} > 1$.
Under \emph{conditional cyclic competition}, on the other hand, some local pairwise interactions can be bistable: both of the axial steady states (corresponding to the carrying capacities of one species and zero density for the other) are locally stable and the final outcome of the local competition will depend on the initial conditions.
Mathematically, assuming bistability occurs for interactions between species 1 and 3 this means that $\alpha_{1,3} > 1$ and $\alpha_{3,1} > 1$. Adding a spatial dimension into the model with conditional cyclic competition should preserve the displacement order $1 > 2 > 3 > 1$ as in the classical cyclic competition.
However, the conditional cyclic competition involves some constraints on the diffusion coefficients~\cite{ADD10}: using arbitrary diffusion coefficients will not guarantee the cyclic dominance.
In this paper, we will explore patterns corresponding to both the classical and the conditional cyclic competition scenarios.

Interestingly, for the model of cyclic competition without space, the long term coexistence of all species is not possible~\cite{ML75}, while coexistence can occur when space is considered~\cite{AM12, C15}. Furthermore, the dimensionality of the spatial domain plays a role in the long term persistence of all species~\cite{AM12}, and various scenarios of coexistence in the model~\eqref{eq:scaledsystem} have been reported in the case of two spatial dimensions. Along with well-known patterns of dynamics such as spiral waves and target waves and their combinations, new regimes have been recently demonstrated including the spread of irregular patches, zip-shaped waves and wedge-shaped traveling fronts~\cite{AM12, C15}. It has also been observed that the species mobilities $\epsilon_{2},\epsilon_{3} \le 1$ can play a key role in ensuring species persistence and provide the system with a rich variety of dynamical regimes~\cite{AM12}. In this work we will demonstrate novel patterns of dynamics for the cyclic Lotka-Volterra reaction-diffusion model with non-equal diffusion coefficients, which have not been reported before.

Finally, the solutions of system~\eqref{eq:scaledsystem} largely depend upon the habitat and initial conditions. In all cases, we consider the spatial domain $\Omega$ to be a $d$-dimensional cube, $d=2,3$, and we explore several different types of initial conditions which imitate different scenarios of biological invasion. In the two-dimensional case, we use the following conditions:
(i) from a point near the centre of the habitat, we extend three boundary lines at angle $2\pi/3$ from each other, thus dividing the habitat into three `triangular' subdomains, and we start with one and only one species present in each of these subdomains at its carrying capacity density (further we call this condition the `triangular' condition for brevity);
(ii) in the habitat initially occupied by a single species (with the density corresponding to the carrying capacity) we introduce the other two species locally in a circular subdomain such that the centres of the two circles do not coincide.
In the three-dimensional case, we explore using initial conditions formed by dividing the whole domain into 8 equal boxes and consider that the boxes are initially occupied by only one species at the density corresponding to the carrying capacity.
In all cases, the initial conditions are smoothed by replacing the zero density/carrying capacity interfaces with a thin layer of smooth transition.

\section{Numerical scheme}\label{sec:numericalscheme}

The numerical scheme we use to approximate solutions to the model~\eqref{eq:scaledsystem} is based on a second order $C^0$-conforming finite element method in space, coupled with a second order Crank-Nicolson Adams-Bashforth implicit-explicit (IMEX) time discretisation.
The spatial mesh is then automatically locally coarsened and refined on certain time steps in response to an estimate of the local contribution to the discretisation error, evaluated using the \emph{error indicator} derived in Section~\ref{sec:apost}.
Without adaptive meshing, the finite element approximation would have to be computed using a uniform mesh which is sufficiently fine everywhere to resolve the (moving) small features appearing in the simulations.
Since solutions of this problem contain a large range of length scales, from the small features of the chaotic interaction regions to the large areas where a single species dominates, the use of a uniform mesh sufficiently fine to resolve all solution features would be prohibitively expensive in terms of memory consumption and computation time, even on high performance hardware.
This issue is exacerbated in three spatial dimensions, and the problem quickly becomes completely intractable with a fine fixed mesh.
Instead, the method can be made significantly more efficient by adaptively refining the mesh where the estimated error is high, adding extra resolution where it is needed, and coarsening where the estimated error is low and less resolution is acceptable.
A simple adaptive algorithm which achieves this is outlined in Algorithm~\ref{alg:FEM}.

We begin by re-writing the model~\eqref{eq:scaledsystem} in the more general matrix form
\begin{align*}
    \ddt{\uvec} - \Delta (\diff \cdot \uvec) &= \f(\uvec)\phantom{0\initialcond(\vx)} \text{ in } \domain \\
    \ddn{\uvec[i]} &= 0\phantom{\initialcond(\vx)\f(\uvec)} \text{ on } \boundary \quad \forall i = 1,2,\dots,n\\
    \uvec(\vx, 0) &= \initialcond(\vx)\phantom{0\f(\uvec)} \text{ in } \domain,
\end{align*}
where $\cdot$ denotes the Euclidean product between vectors.
The nonlinear interaction term $\f : \Re^{m} \to \Re^{m}$ is assumed to satisfy the general Lipschitz-style growth condition that there exists $0 \leq \gamma < 2$ when $\spacedim=2$, or $0 \leq \gamma \le  \frac{4}{3}$ when $\spacedim = 3$, such that for any $\vvec, \wvec \in \Re^{m}$ we have
\begin{align}\label{eq:growthBound}
    \abs{\f(\vvec) - \f(\wvec)} \leq \Cf (1 + \abs{\vvec}^{\gamma} + \abs{\wvec}^{\gamma}) \abs{\vvec - \wvec},
\end{align}
with $\abs{\cdot}$ denoting the Euclidean norm on $\Re^{m}$.

The interaction term of the model~\eqref{eq:scaledsystem} may be seen to satisfy this condition with $\gamma = 1$ by writing it as
\begin{align*}
	\f(\uvec) = (\diag{\growth} \uvec) \elprod (\vec{1} - A \uvec),
\end{align*}
where $\vec{1} \in \Re^{m}$ denotes the vector with all entries equal to one, the matrix $A \in \Re^{m \times m}$ contains the interaction parameters, with $A_{i,j} = \reac{i}{j}$, and $\elprod$ denotes the Hadamard product between tensors, and for $r \in \Re^{m}$ we use $\diag{r}$ to denote the $m \times m$ matrix with the elements of the vector $r$ on its main diagonal.

We shall use the notation $H^1(\domain)$ to denote the Hilbertian Sobolev space of functions with square-integrable first derivatives, viz.
\begin{align*}
H^1(\domain) = \{ v \in L^2(\domain) : \nabla v \in (L^2(\domain))^m \},
\end{align*}
and we refer to, e.g., \cite{AF03,W87} for more on Sobolev and Bochner functions spaces.
Multiplying by a suitable test function and integrating by parts, the problem may be written in weak form as: find $\uvec \in H^1(0,T; H^1(\domain)^{m})$ such that
\begin{align}\label{eq:weakForm}
    \left( \ddt{\uvec}, \vvec \right) + (\diff \jacobian{\uvec}, \jacobian{\vvec}) - (\f(\uvec), \vvec) = 0 \qquad \forall \vvec \in H^1(\domain)^{m}.
\end{align}
Here, $\jacobian{\uvec}$ denotes the Jacobian matrix of $\uvec$ and the $L^2$ inner product between matrix-valued functions $A(\vx)$ and $B(\vx)$ is defined to be $(A, B) := \int_{\domain} A : B \dx$, where $A:B$ denotes the Frobenius product.

The spatially discrete finite element method for solving this problem can then be written as: find $\uhvec \in H^1(0,T; (\Vh)^m)$ such that
\begin{align}\label{eq:finiteElementForm}
    \left( \uhtvec, \vectestfn \right) + (\diff \jacobian{\uhvec}, \jacobian{\vectestfn}) - (\f(\uhvec), \vectestfn) = 0 \qquad \forall \vectestfn \in (\Vh)^m.
\end{align}
Here, $\Vh$ is a conforming finite element space, that is a finite dimensional space of functions which are continuous over $\domain$ and piecewise polynomial of degree $\k$ with respect to a mesh $\Th$ covering $\domain$, see e.g. the monograph~\cite{C}. In particular, we select piecewise quadratic finite elements for our computations.

To produce a practical method, the time derivative in the spatially discrete problem presented above must also be discretised. Also, the finite element space on each timestep will be different, in general, although we refrain for expressing this dependence explicitly in the notation for accessibility.
Adopting the second order Crank--Nicolson Adams--Bashforth IMEX time discretisation allows us to treat the linear part of the differential operator implicitly while the nonlinear reaction term is treated explicitly, producing the fully discrete problem: for $1 \leq n \leq N-1$, find $\uvec^{n+1} \in (\Vh)^{m}$ such that
\begin{align*}
    \left(\frac{\uvec^{\timestep + 1} - \uvec^{\timestep}}{\tau}, \vectestfn \right) +  \left( \frac{\diff}{2} \left( \jacobian \uvec^{\timestep + 1} + \jacobian \uvec^{\timestep} \right), \jacobian \vectestfn \right) = \left( \frac{3}{2} \f(\uvec^{\timestep}) - \frac{1}{2} \f(\uvec^{\timestep-1}), \vectestfn \right) \qquad \forall \vectestfn \in (\Vh)^{m},
\end{align*}
where the function $\uvec^{n}$ is the finite element approximation to the solution $\uvec$ at time moment $t^n$.
This choice of timestepping scheme avoids the need to solve a system of nonlinear equations at each time step while still providing second-order accuracy in time. Moreover, combined with our choice of quadratic finite elements, this choice has proven to be a good compromise between accuracy and computational cost when combined with the adaptive scheme described in Algorithm~\ref{alg:FEM}.

\begin{algorithm}
  \caption{A typical adaptive finite element method}
  \label{alg:FEM}
  \begin{algorithmic}[1]
    \Require Domain, initial condition, initial mesh
    \State Interpolate initial condition
    \State Assemble system matrices on initial mesh
    \For{each time step}
    	\State Compute nonlinear forcing term using previous solution
		\State Solve for current time step
		\If{mesh should be adapted on this step}
    		\State Estimate error
    		\State Refine/coarsen mesh
    		\State Transfer solution to new mesh
			\State Recompute system matrices on new mesh
		\EndIf
    \EndFor
    \Ensure Discrete solution at each time step
  \end{algorithmic}
\end{algorithm}

The above numerical method was implemented using the \texttt{deal.II} C++ finite element library~\cite{BHK07}.
An essential feature of this library is its support for high order spatial discretisations with adaptive mesh refinement  in both two and three spatial dimensions.

\section{New patterns of spread in the cyclic competition model}\label{sec:results}

We start our investigation of~\eqref{eq:scaledsystem} with the case $d=2$. Our goal is to report some previously unknown transient regimes in which the area where all three species coexist spreads into areas where only one of the species is present. We first consider the classical cyclic competition scenario with model parameters given by $\alpha_{1,1} = 1$; $\alpha_{1,2} = 1$; $\alpha_{1,3} = 2$; $\alpha_{2,1} = 2$; $\alpha_{2,2} = 1$; $\alpha_{2,3} = 1$; $\alpha_{3,1} = 1$; $\alpha_{3,2} = 2$; $\alpha_{3,3} = 1$.
Our numerical simulations revealed new dynamical patterns with regular spatial structures in the travelling population waves. The geometrical shape of the regular structures in the wake of the front largely depends on the diffusion coefficients.

\begin{figure}
    \subcaptionbox{$t = 50$}[.5\textwidth]{%
        \includegraphics[width=0.49\textwidth]{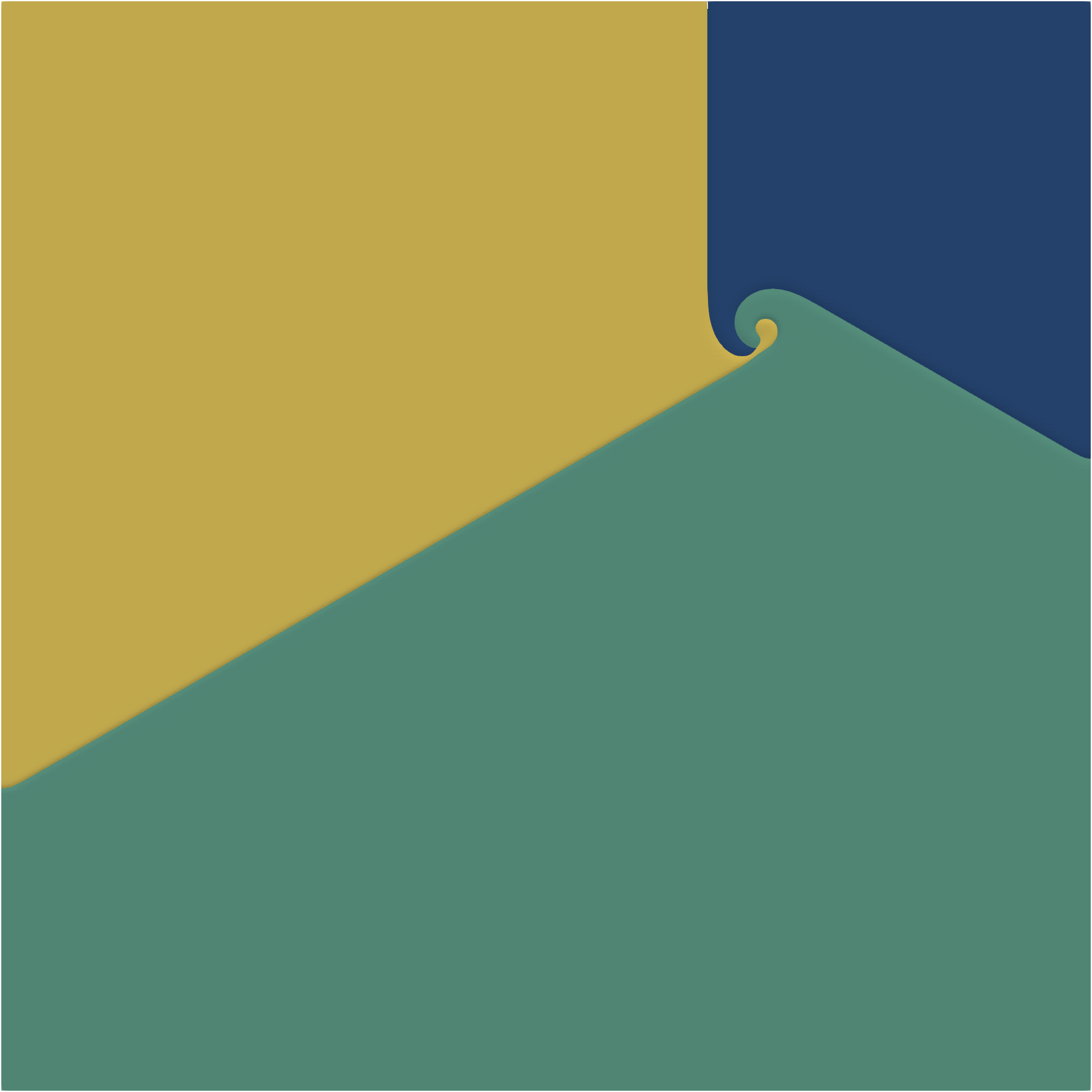}
    }\hfill%
    \subcaptionbox{$t = 150$}[.5\textwidth]{%
        \includegraphics[width=0.49\textwidth]{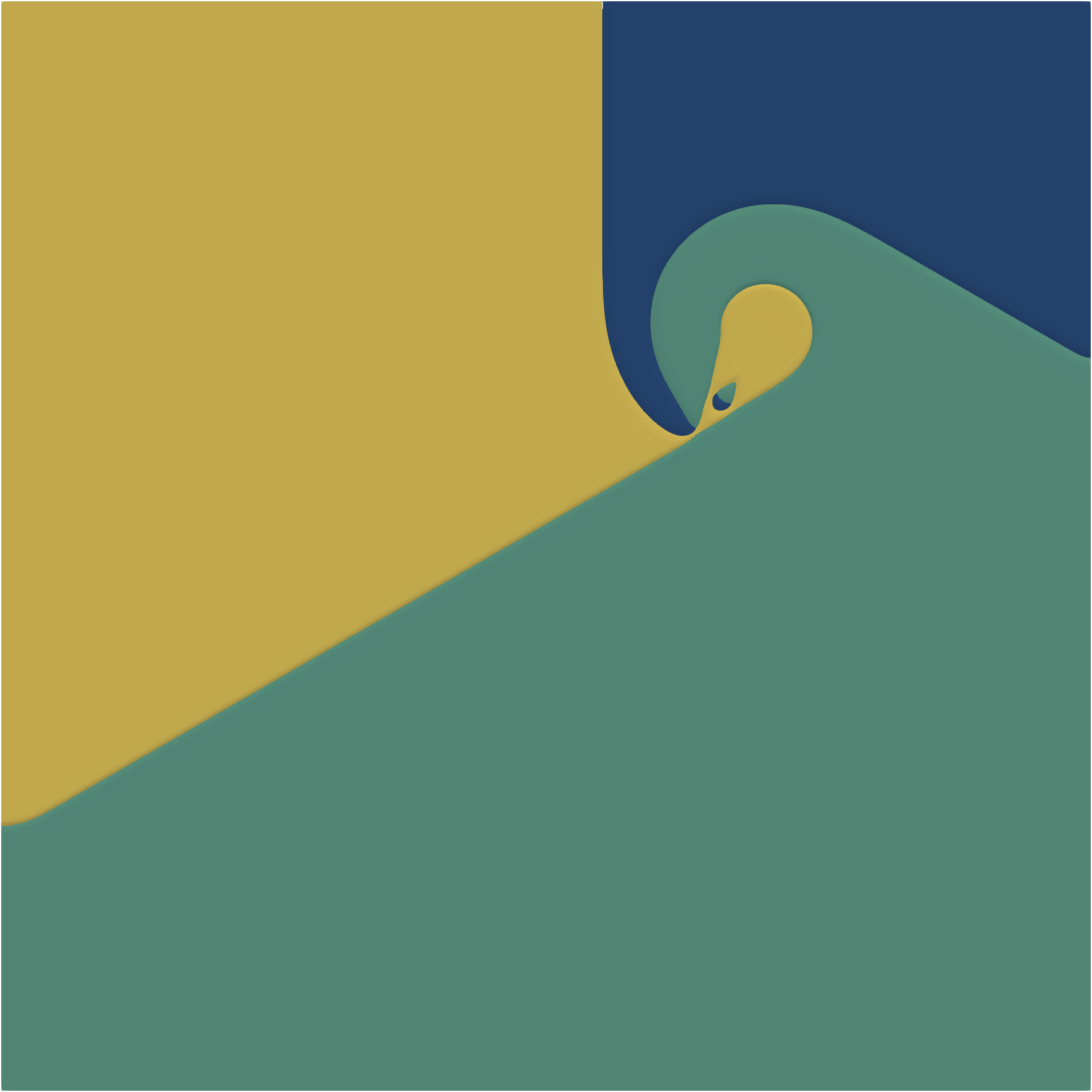}
    }
    \\[1em]
    \subcaptionbox{$t = 360$}[.5\textwidth]{%
        \includegraphics[width=0.49\textwidth]{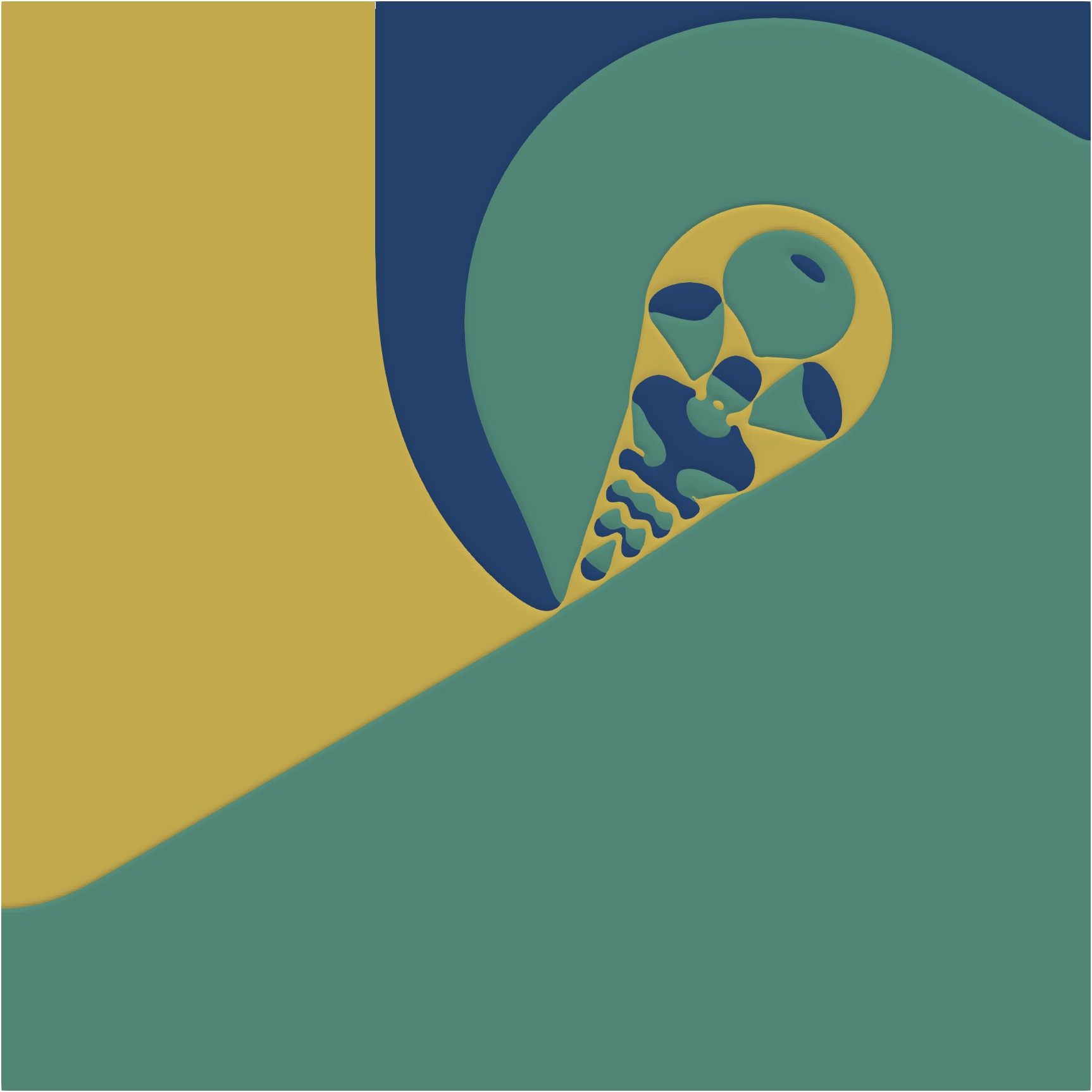}
    }\hfill%
    \subcaptionbox{$t = 900$}[.5\textwidth]{%
        \includegraphics[width=0.49\textwidth]{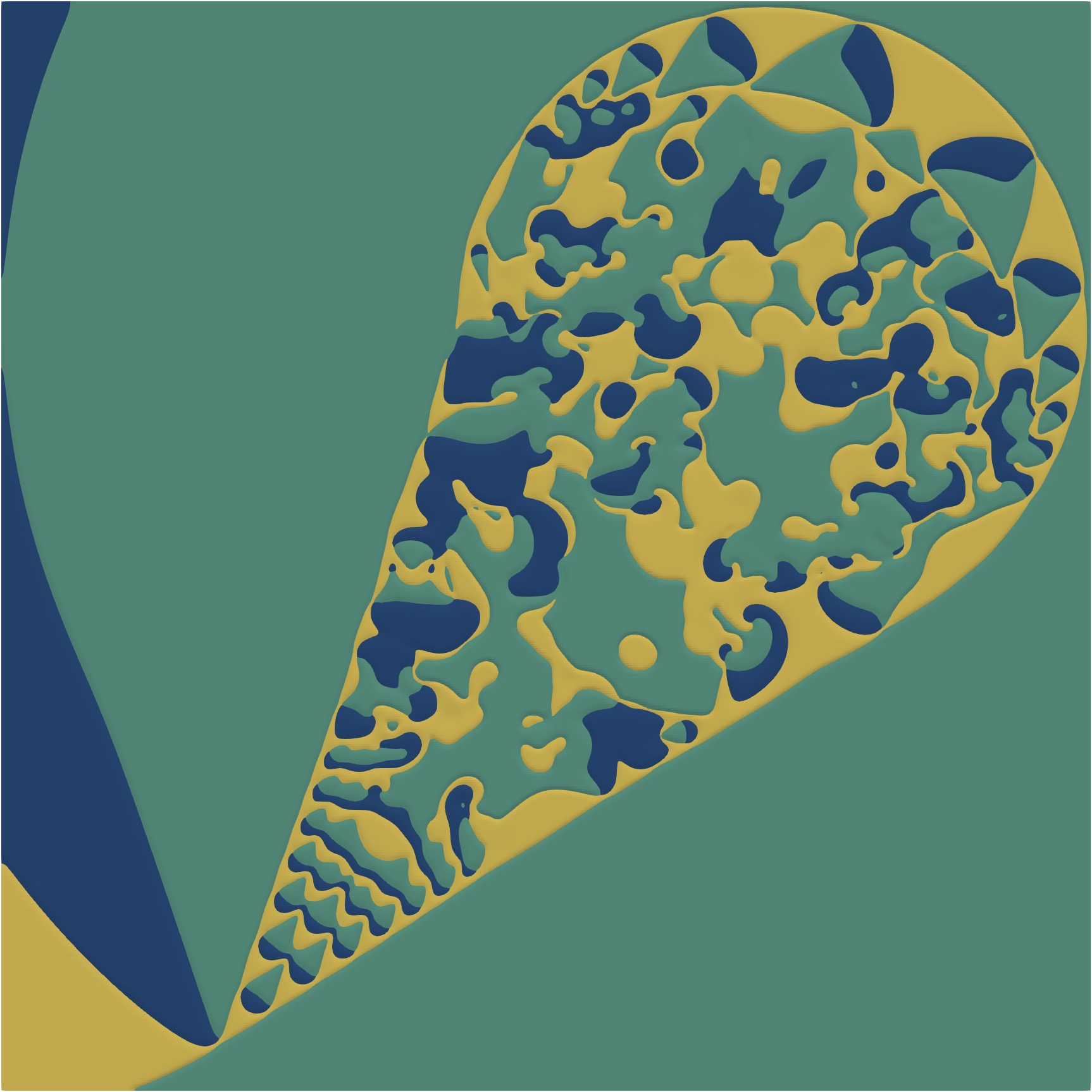}
    }
    \caption{The evolution of the `triangular droplet-like' patterns which are observed when $\epsilon_{2}=0.1$ and $\epsilon_{3}=0.6$ using the `triangular' initial conditions.}
    \label{fig:cycliccompetition:icecream}
\end{figure}

\begin{figure}
    \subcaptionbox{$t = 0$\label{fig:patchInitialConditions}}[.5\textwidth]{%
        \includegraphics[width=0.4\textwidth]{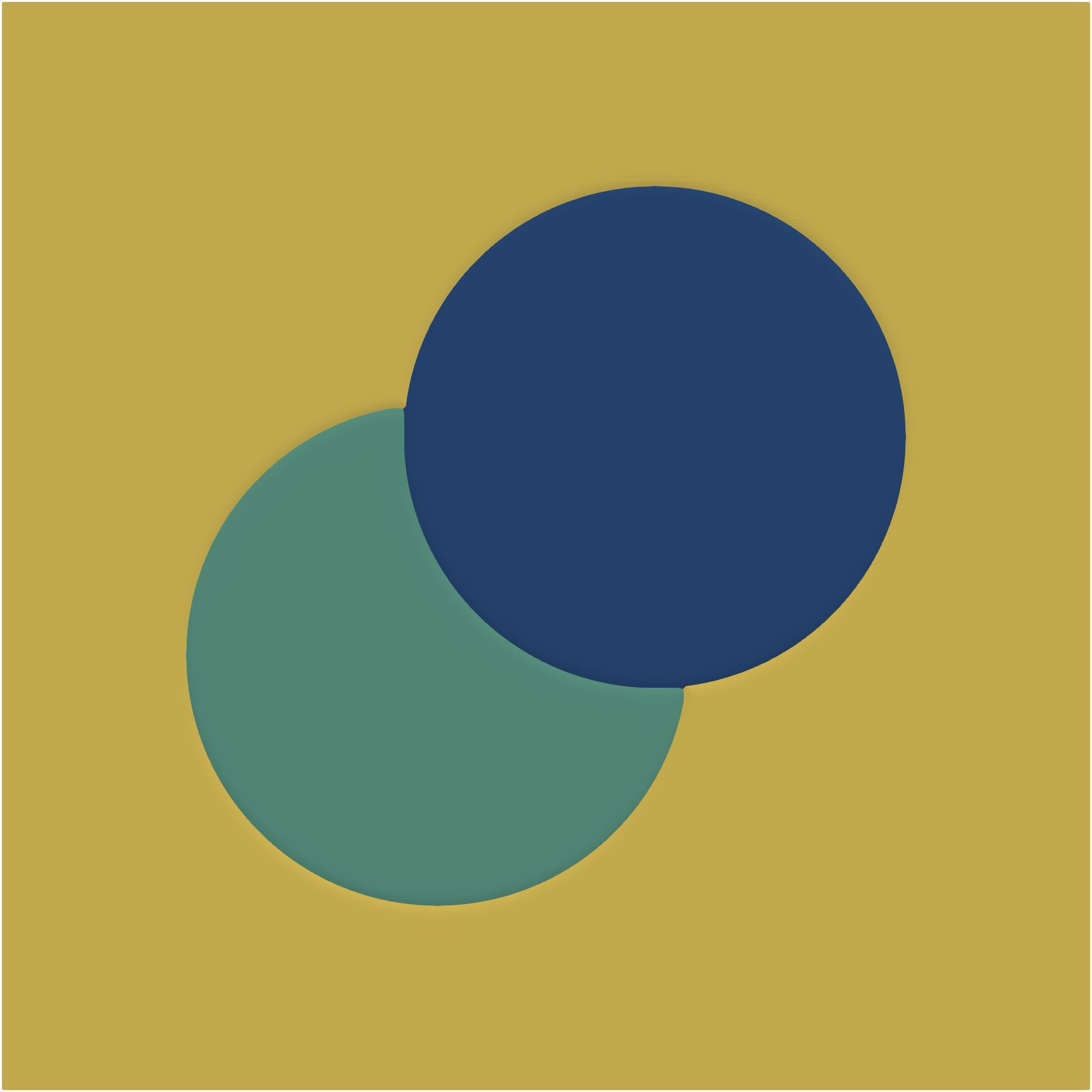}
    }\hfill%
    \subcaptionbox{$t = 40$}[.5\textwidth]{%
        \includegraphics[width=0.4\textwidth]{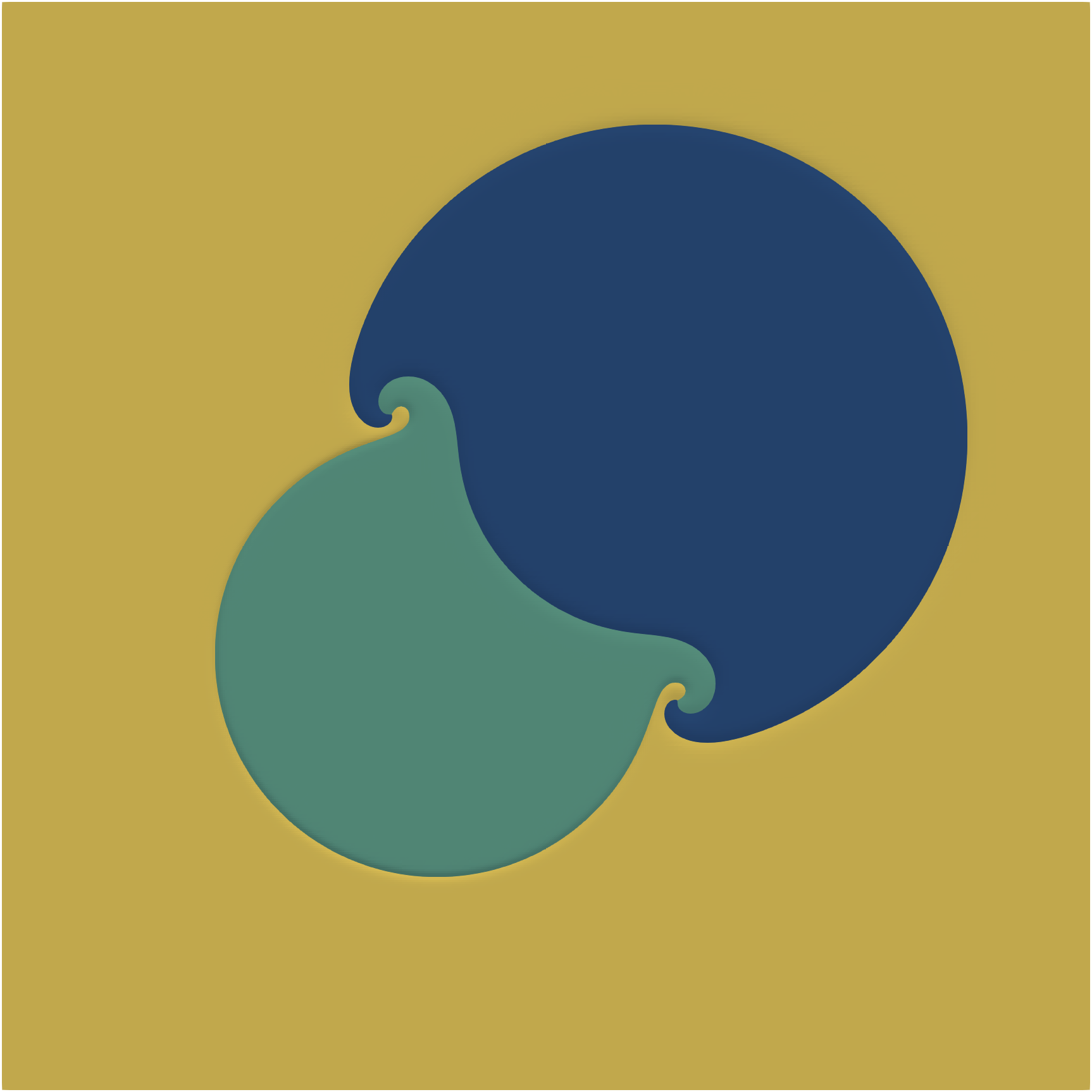}
    }
    \\[0em]
    \subcaptionbox{$t = 120$}[.5\textwidth]{%
        \includegraphics[width=0.4\textwidth]{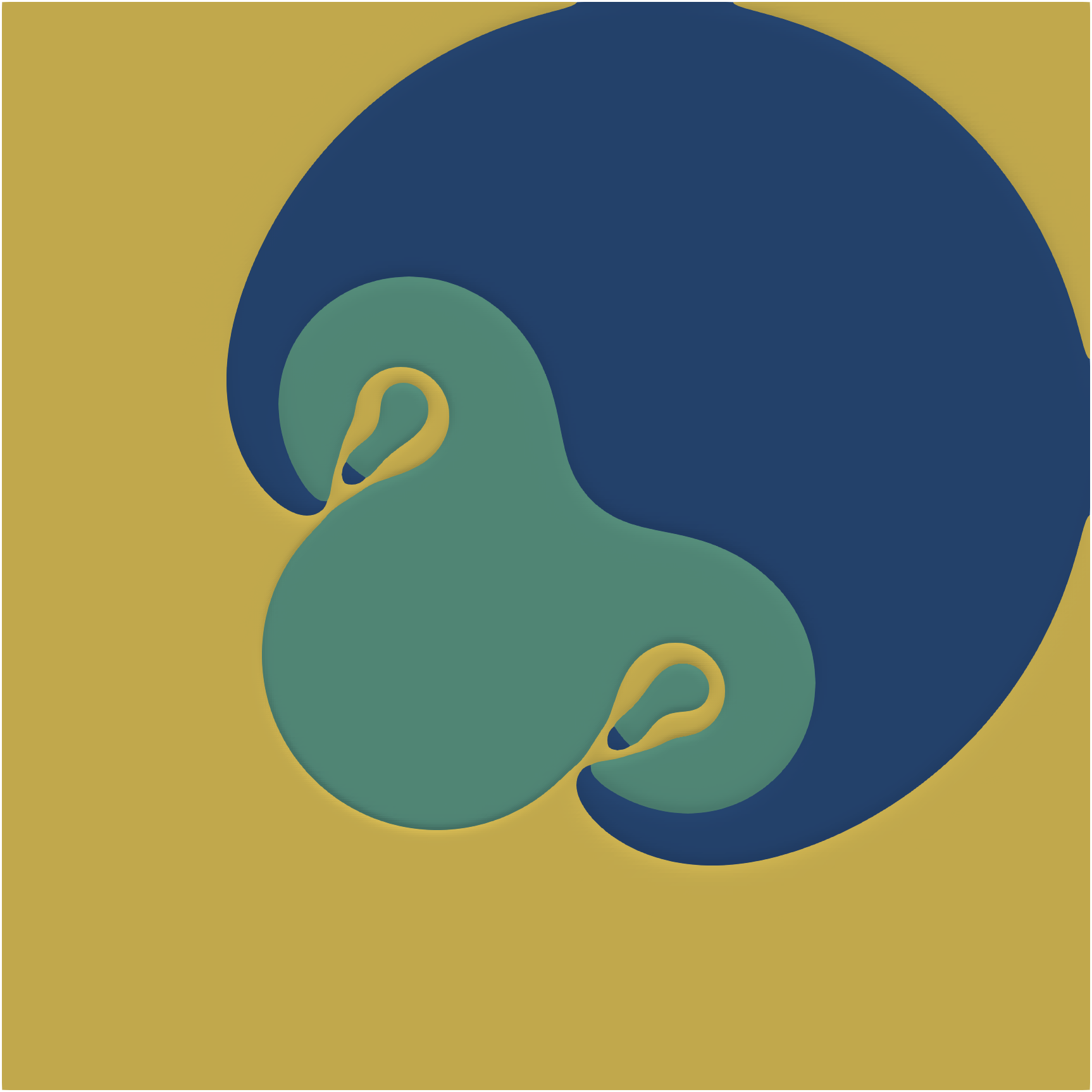}
    }\hfill%
    \subcaptionbox{$t = 240$\label{fig:circlepatches_240}}[.5\textwidth]{%
        \includegraphics[width=0.4\textwidth]{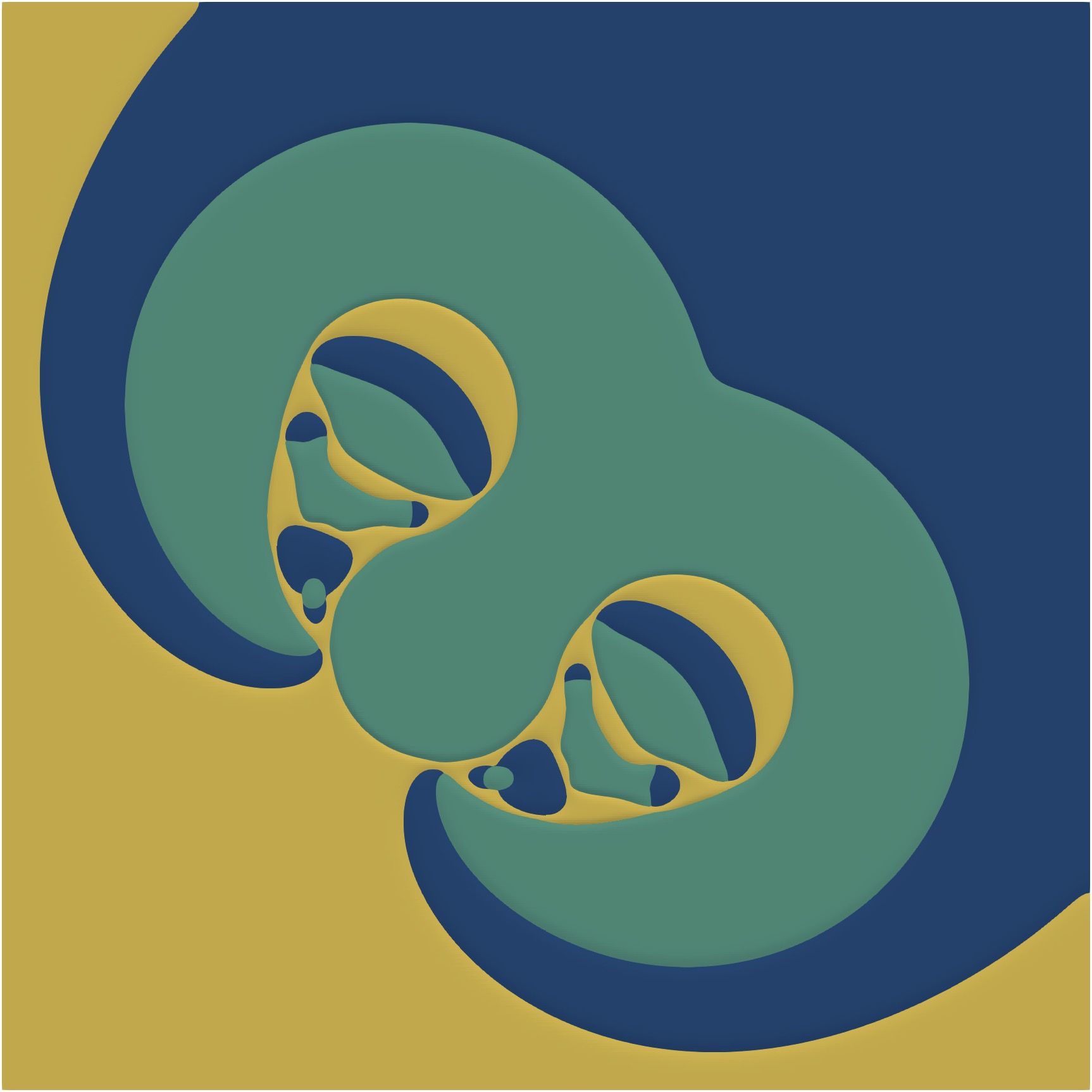}
    }
    \\[0em]
    \subcaptionbox{$t = 400$}[.5\textwidth]{%
        \includegraphics[width=0.4\textwidth]{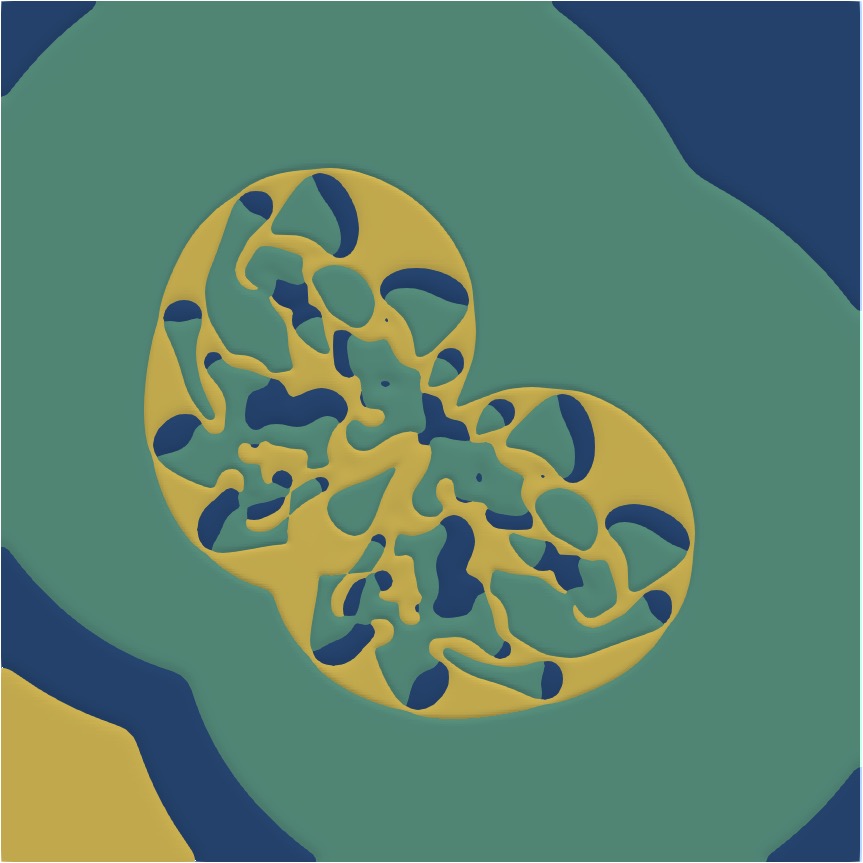}
    }\hfill%
    \subcaptionbox{$t = 600$}[.5\textwidth]{%
        \includegraphics[width=0.4\textwidth]{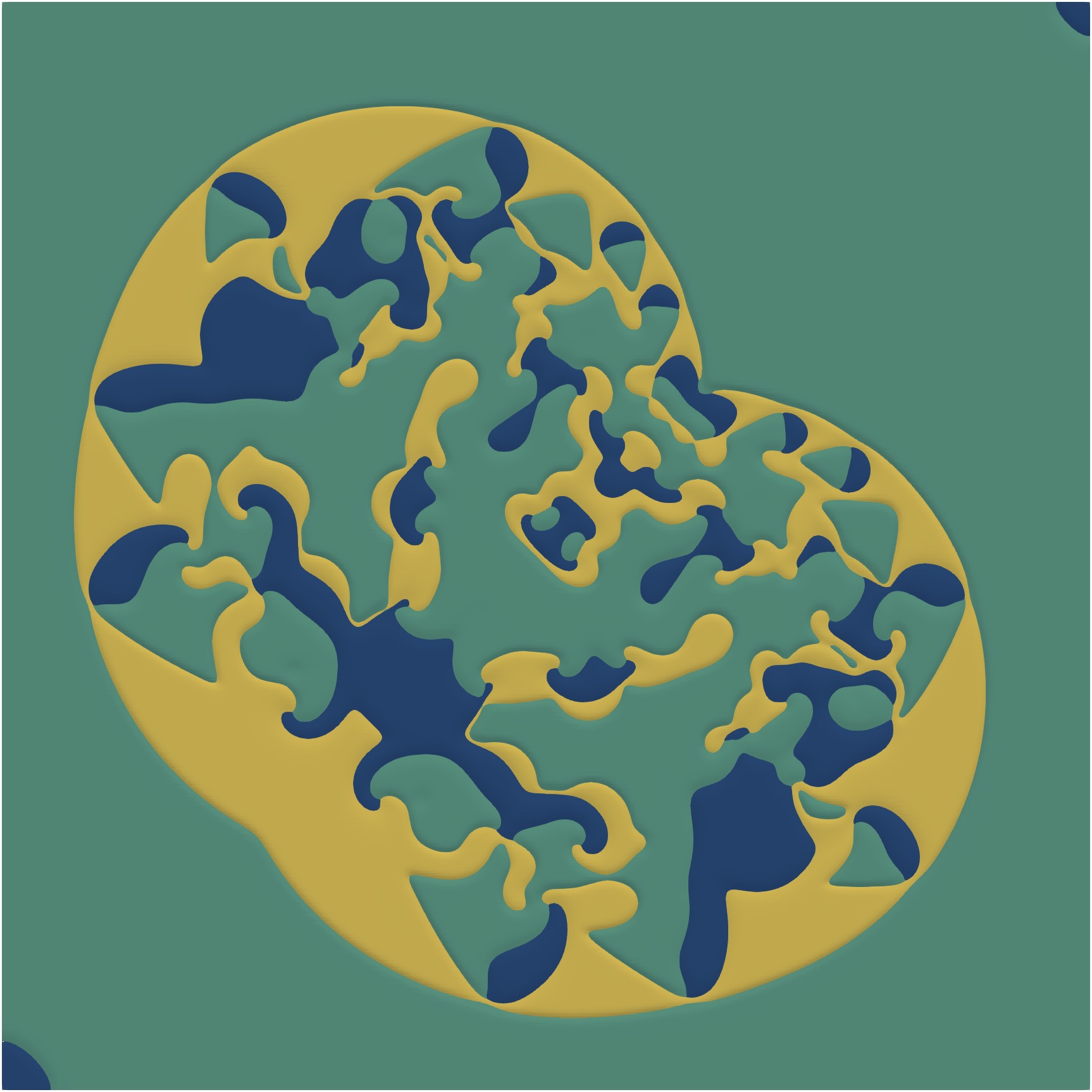}
    }
    \caption{The evolution of the `triangular droplet-like' patterns with $\epsilon_{2}=0.1$ and $\epsilon_{3}=0.6$ using the initial conditions shown in Fig.~\ref{fig:patchInitialConditions}.}
    \label{fig:cycliccompetition:icecream2}
\end{figure}

For small values of the diffusion coefficient of the second species, the resultant patterns have a triangular droplet-like shape which is shown in Fig.~\ref{fig:cycliccompetition:icecream}.
This corresponds to the diffusion coefficients $\epsilon_{2}=0.1$ and $\epsilon_{3}=0.6$. In this and all subsequent figures we use the following colour coding: dark blue indicates regions dominated by species $\uvec[1]$, light yellow indicates regions dominated by species $\uvec[2]$, mid green indicates regions dominated by species $\uvec[3]$.

We use the symmetrical triangular initial conditions described in Section~\ref{sec:model}. During the initial stages of the simulation, a spiral tip is formed. This subsequently transforms into a symmetrical droplet-like domain with a sharp wedge and a complex structure inside. The sharp wedge moves at a constant speed towards the low left corner of the domain. At the tip of the wedge, all three species are present. A regular structure emerges at the front of the wedge, and both move with the same speed. This regular structure is formed of small droplet-like units which are periodic in the direction orthogonal to the bisector of the wedge.
The number of small droplets in the wedge is constantly growing through time. However, moving  progressively from the wedge to the centre of the spreading domain in which the species coexist, the spatial structure becomes more irregular, thus the domain of regularity in the species spread is transient. At much larger times, after the spreading wedge as well as the other spreading boundaries eventually hit the boundary of the habitat $\Omega$, the dynamics of the patches eventually becomes chaotic both in space and time (we do not show this pattern for the sake of simplicity).

We explored the robustness of the pattern shown in Fig.~\ref{fig:cycliccompetition:icecream} with respect to the choice of initial conditions.
We found that similar patterns can be observed for the various initial conditions mentioned in Section~\ref{sec:model} and, in particular, when the initial angles of the sectors of species distributions are different and when the species are initially located in six sectors instead of three (results not shown).
In other cases, such as when the three species are initially contained within squares or disks, we found that spreading wedges generating regular droplet-like structures are once again produced.
However, with highly symmetric initial arrangements, these spreading regions of regularity can eventually annihilate one another by colliding. Fig.~\ref{fig:cycliccompetition:icecream2} shows an example of the evolution of the droplet-like structures for the initial conditions given by circular subdomains, cf. Fig.~\ref{fig:patchInitialConditions}.

Next, we investigated the dependence of patterns of spreading waves on the diffusion coefficients $\epsilon_{2}$ and $\epsilon_{3}$, using the same initial conditions as for Fig.~\ref{fig:cycliccompetition:icecream}.
The results of this investigation are summarised in Fig.~\ref{fig:parameterSearch}.
Decreasing the diffusion coefficient of the second competitor to $\epsilon_{2}=0.05$ appears to prevent the formation of regular droplets in the wedge of invasion (see Fig.~\ref{fig:e2=0.05_e3=0.55}).
The spreading sharp wedge of invasion, however, continues to persist in this case.
As $\epsilon_2$ is increased, however, we find that the regular pattern in the wake of the front persists until $\epsilon_{2}=0.35$; for higher $\epsilon_{2}$, the spatial structure in the spreading wedge becomes irregular (see Fig.~\ref{fig:e2=0.35_e3=0.55}). When the mobility of the third competitor is low (e.g. for $\epsilon_{3}=0.35$), only a single central droplet structure emerges in the spreading wedge (see Fig.~\ref{fig:e2=0.1_e3=0.35}). Finally, with increased $\epsilon_{3}$, the angle of the spreading wedge abruptly increases and a new pattern of dynamics appears (see the next paragraph). This happens for $\epsilon_{3}$ close to one. Finally, for all diffusion coefficients very close to one, the pattern of spread consists of spiral waves (see Fig.~\ref{fig:e2=1_e3=1}) which are well known for this system.

\begin{figure}
    \subcaptionbox{$\epsilon_2 = 0.05$, $\epsilon_3 = 0.55$\label{fig:e2=0.05_e3=0.55}}[.5\textwidth]{%
		\includegraphics[width=.49\textwidth]{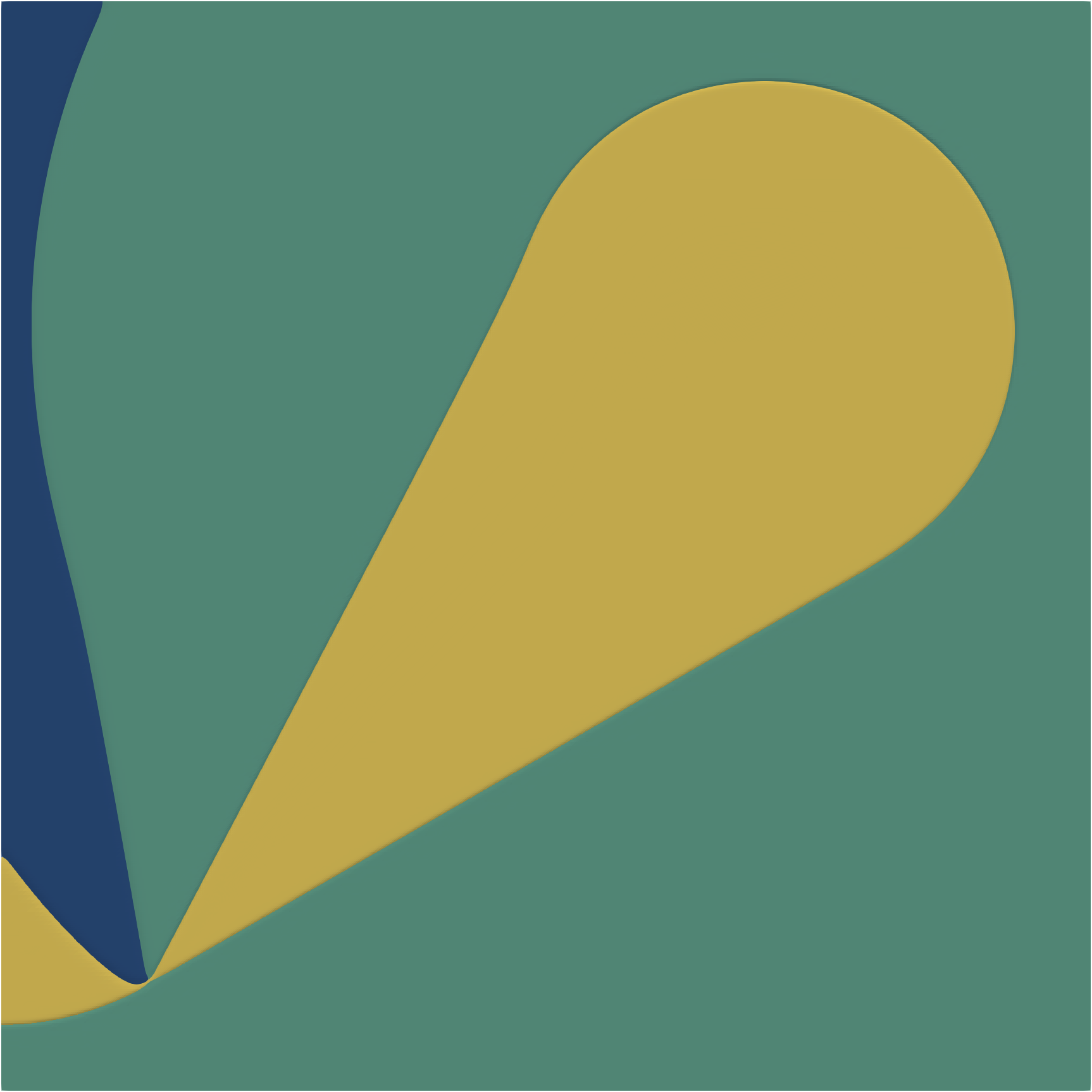}
    }%
    \subcaptionbox{$\epsilon_2 = 0.35$, $\epsilon_3 = 0.55$ (shown at $t=560$)\label{fig:e2=0.35_e3=0.55}}[.5\textwidth]{%
		\includegraphics[width=.49\textwidth]{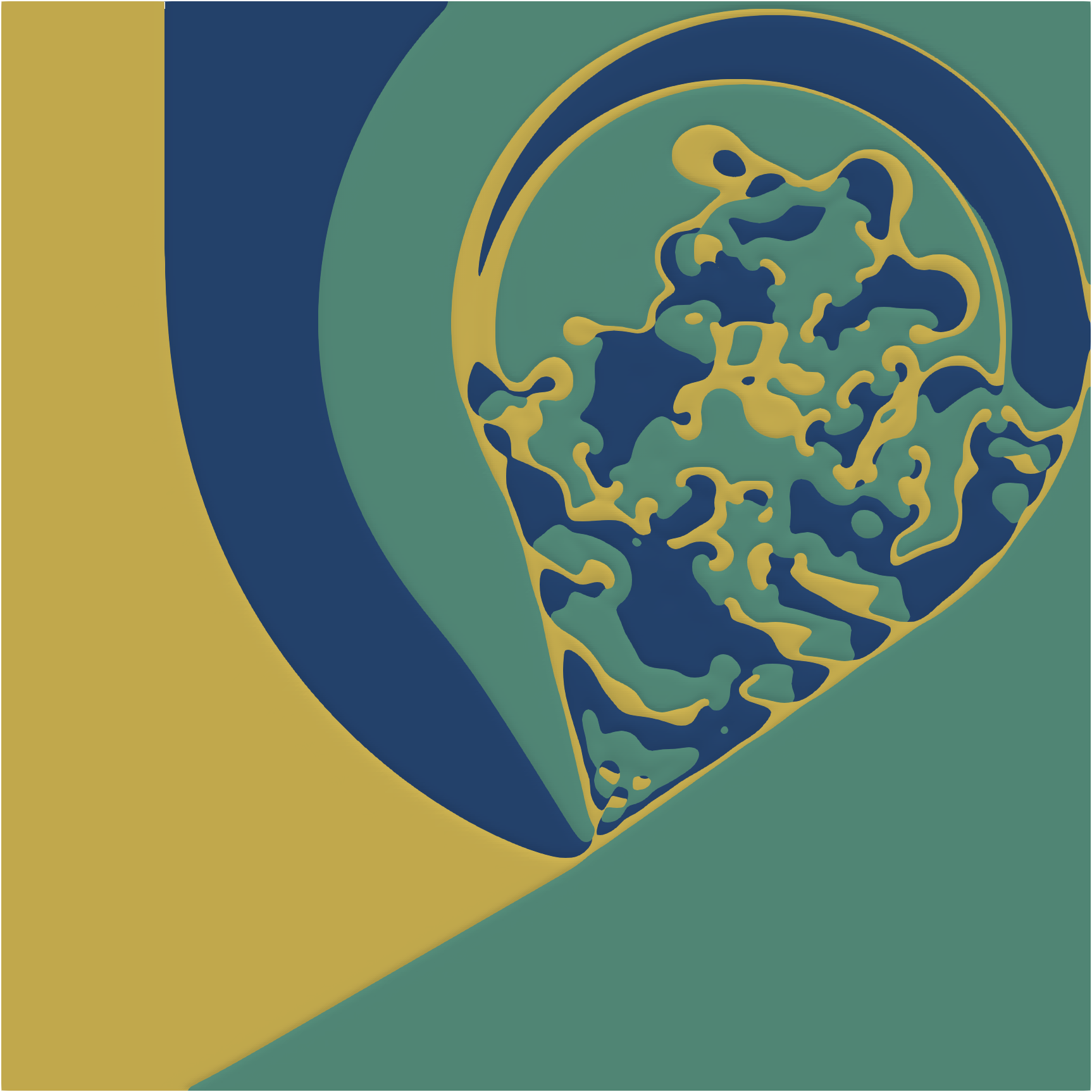}
    }
    \\[1em]
    \subcaptionbox{$\epsilon_2 = 0.1$, $\epsilon_3 = 0.35$\label{fig:e2=0.1_e3=0.35}}[.5\textwidth]{%
		\includegraphics[width=.49\textwidth]{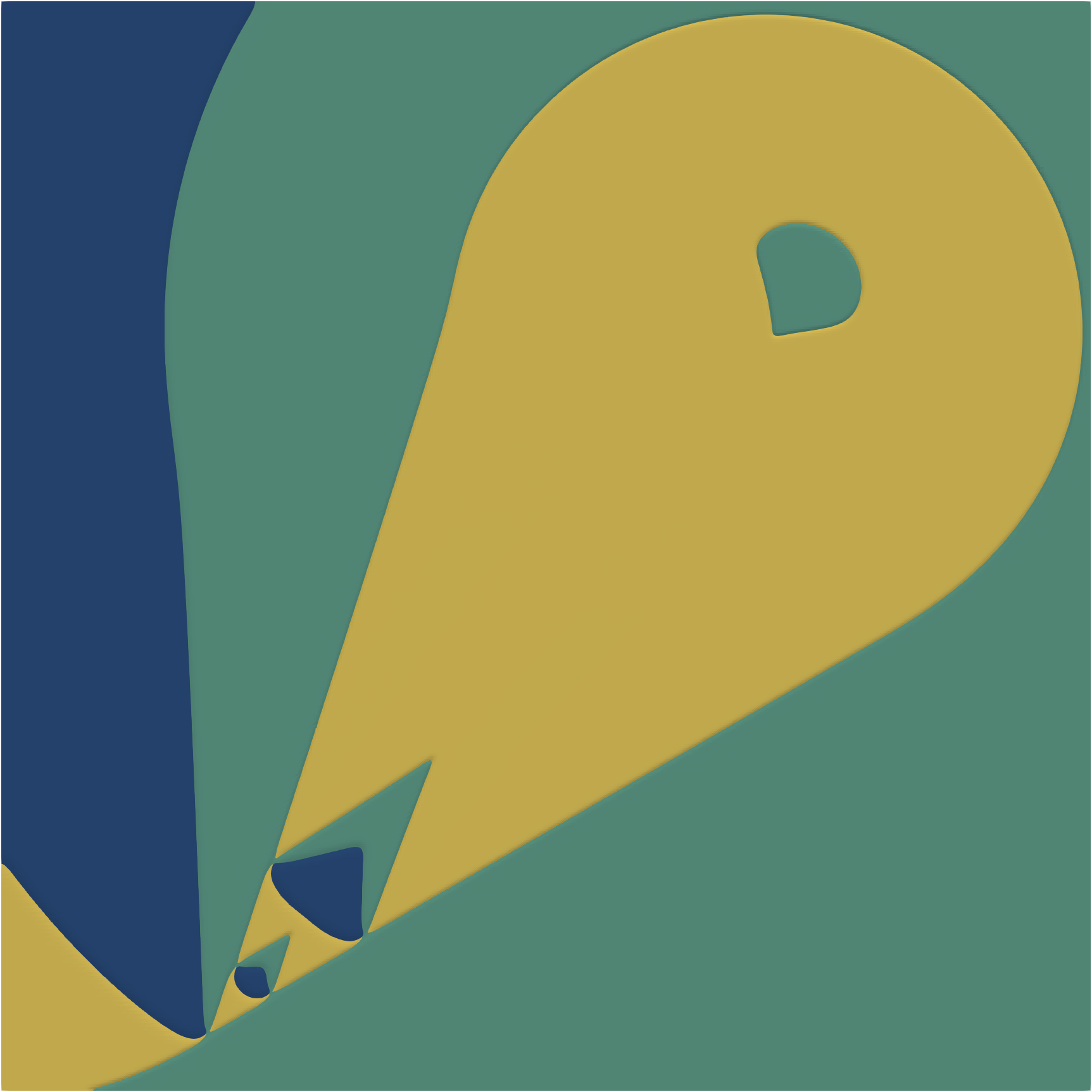}
    }%
    \subcaptionbox{$\epsilon_2 = \epsilon_3 = 1$\label{fig:e2=1_e3=1}}[.5\textwidth]{%
		\includegraphics[width=.49\textwidth]{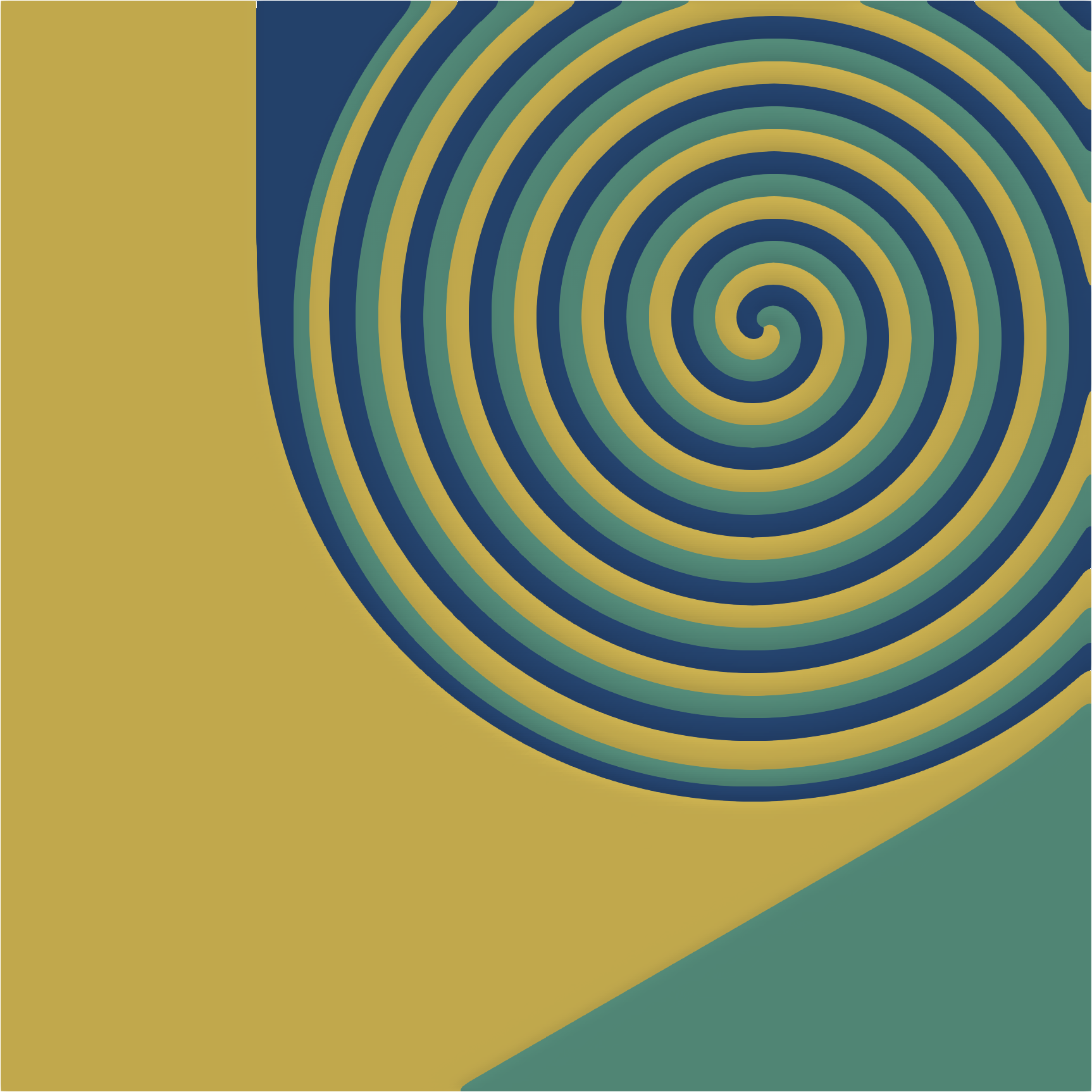}
    }
    \caption{The formation of droplet patterns for various values of the diffusion parameters $\diff[2]$ and $\diff[3]$. Solutions are plotted for $t=900$, unless otherwise stated, having evolved from the same initial conditions shown as in Fig.~\ref{fig:cycliccompetition:icecream}.}
    \label{fig:parameterSearch}
\end{figure}

Fig.~\ref{fig:stripes} demonstrates an example of the species spreading with a stripe-like structure, and corresponds to $\epsilon_{2}=0.1$ and $\epsilon_{3}=0.9$ with the initial condition shown as for Fig.~\ref{fig:cycliccompetition:icecream}, but for a larger spatial domain $\Omega$.
In this case, we observe that the invasion of the domain of coexistence occurs via the formation of a train of parallel bands which move towards the left-hand side boundary. Each band has the same width and is slightly round at the front. Behind the bands, the species distribution is highly irregular.
The spread of the domain of coexistence in the opposite direction (i.e. towards the right boundary) occurs in a different way, via the prorogation of irregular patches, which we refer to as a wave of chaos. Thus, there is a clear anisotropy in terms of patterns of the population spread depending on the direction, which in turn is determined by the initial condition. The domain occupied by the irregular patches eventually grows in size until it invades the whole domain once the bands move out of the domain. One can also see that regular droplet-like patterns travel around the edge of the chaotic domain.

\begin{figure}
    \hfill\subcaptionbox{$t = 40$}[0.495\textwidth]{%
        \includegraphics[width=0.4\textwidth]{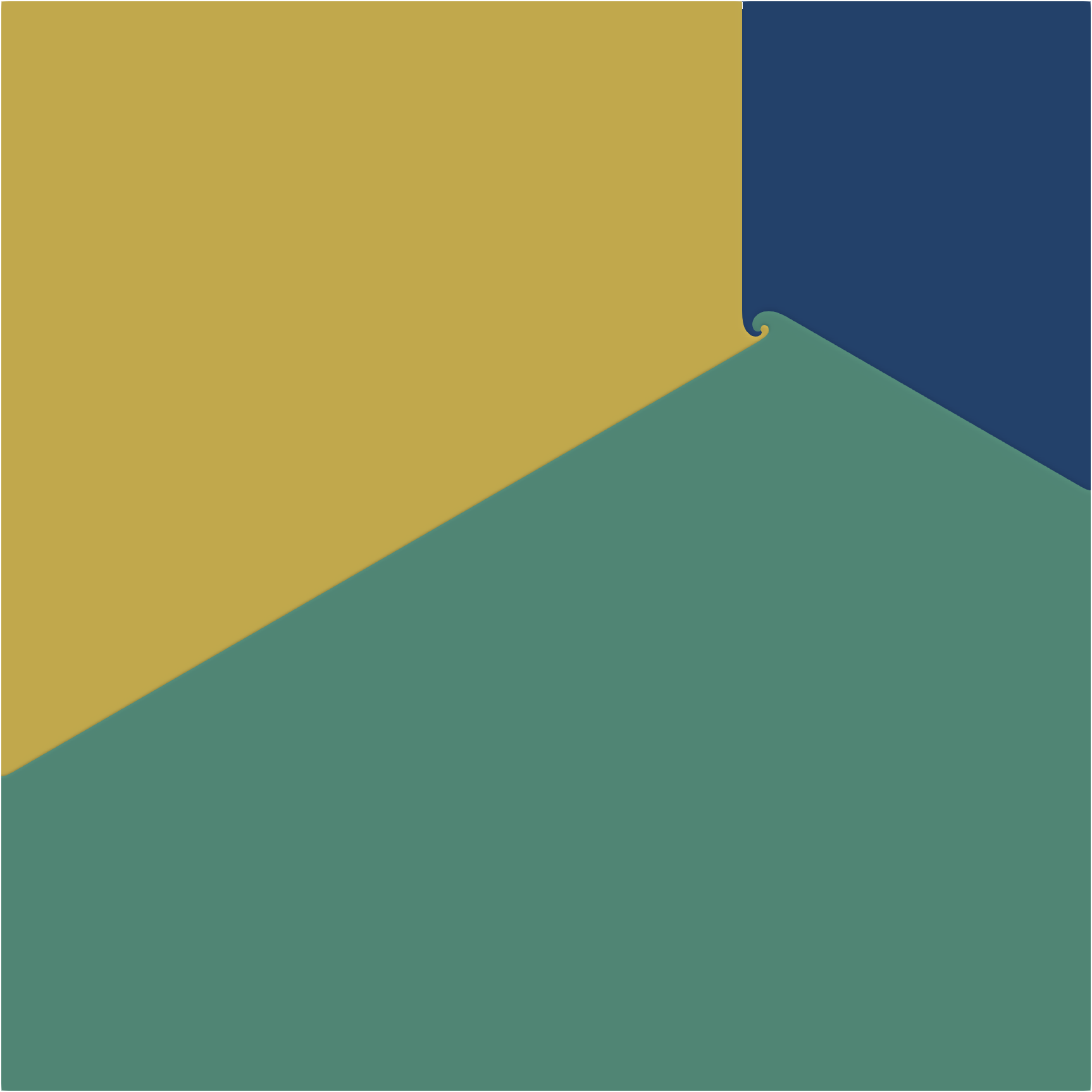}
    }\hfil%
    \subcaptionbox{$t = 160$}[0.495\textwidth]{%
        \includegraphics[width=0.4\textwidth]{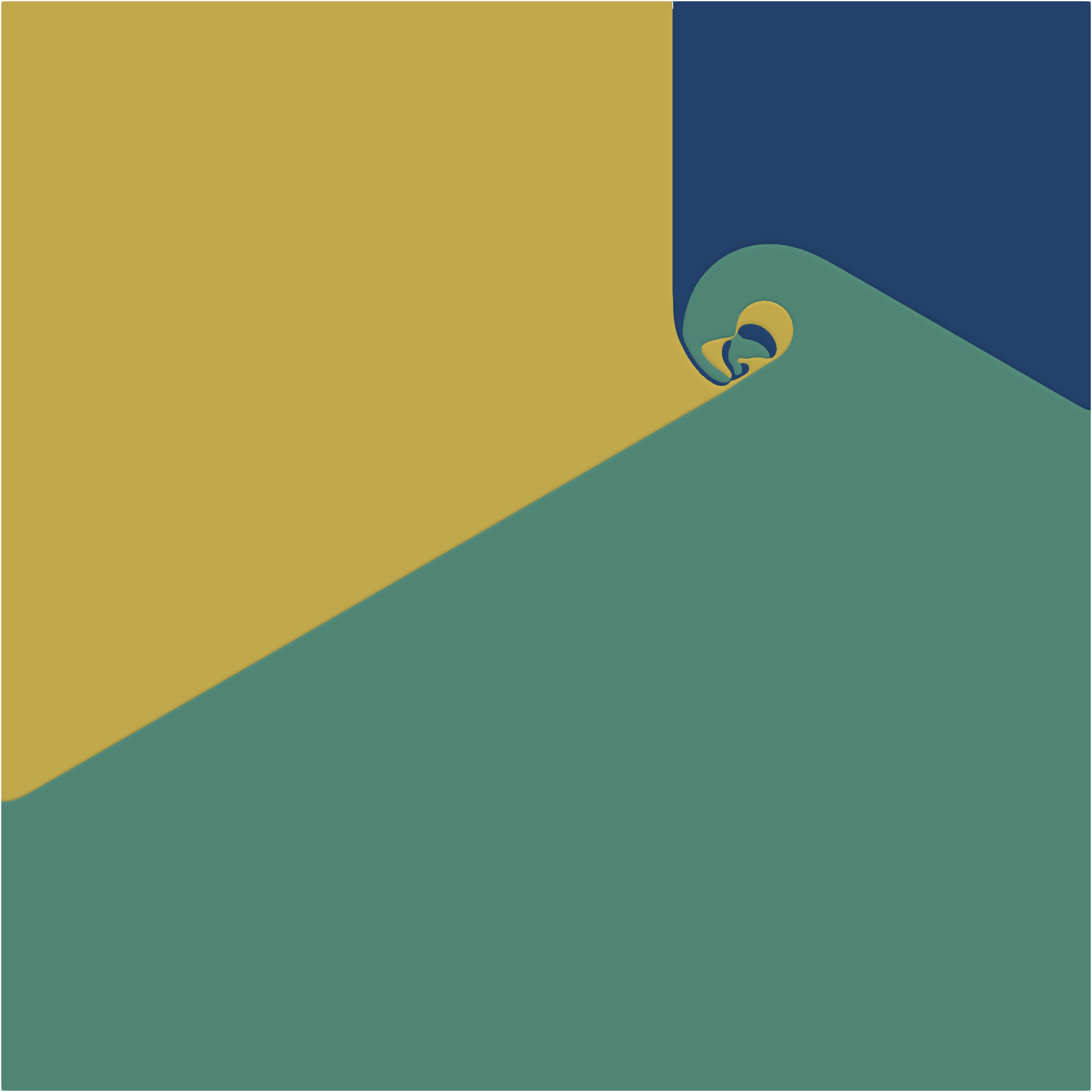}
    }\hfill
    \\[0.0cm]
    \hfill\subcaptionbox{$t = 320$}[0.495\textwidth]{%
        \includegraphics[width=0.4\textwidth]{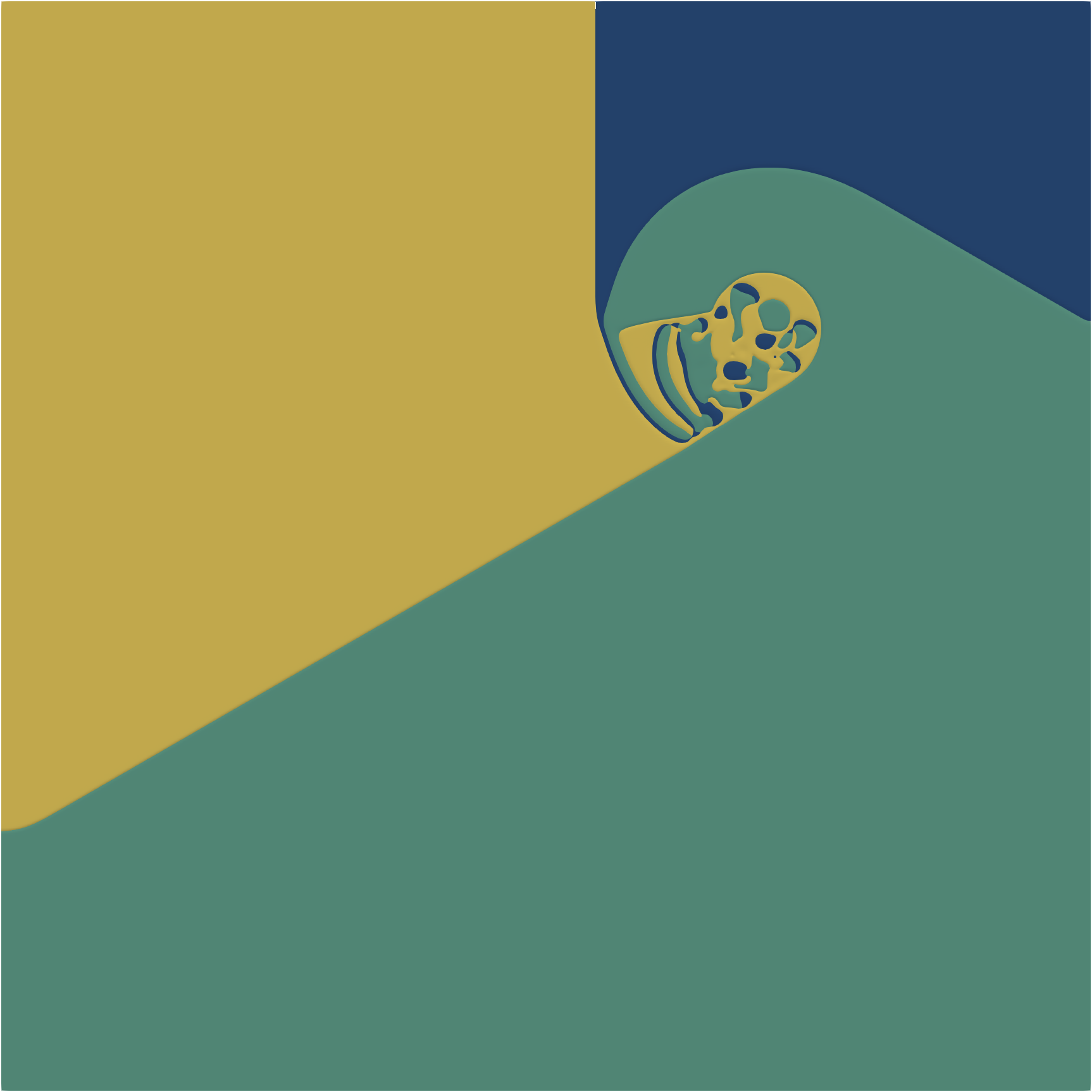}
    }\hfil%
    \subcaptionbox{$t = 640$}[0.495\textwidth]{%
        \includegraphics[width=0.4\textwidth]{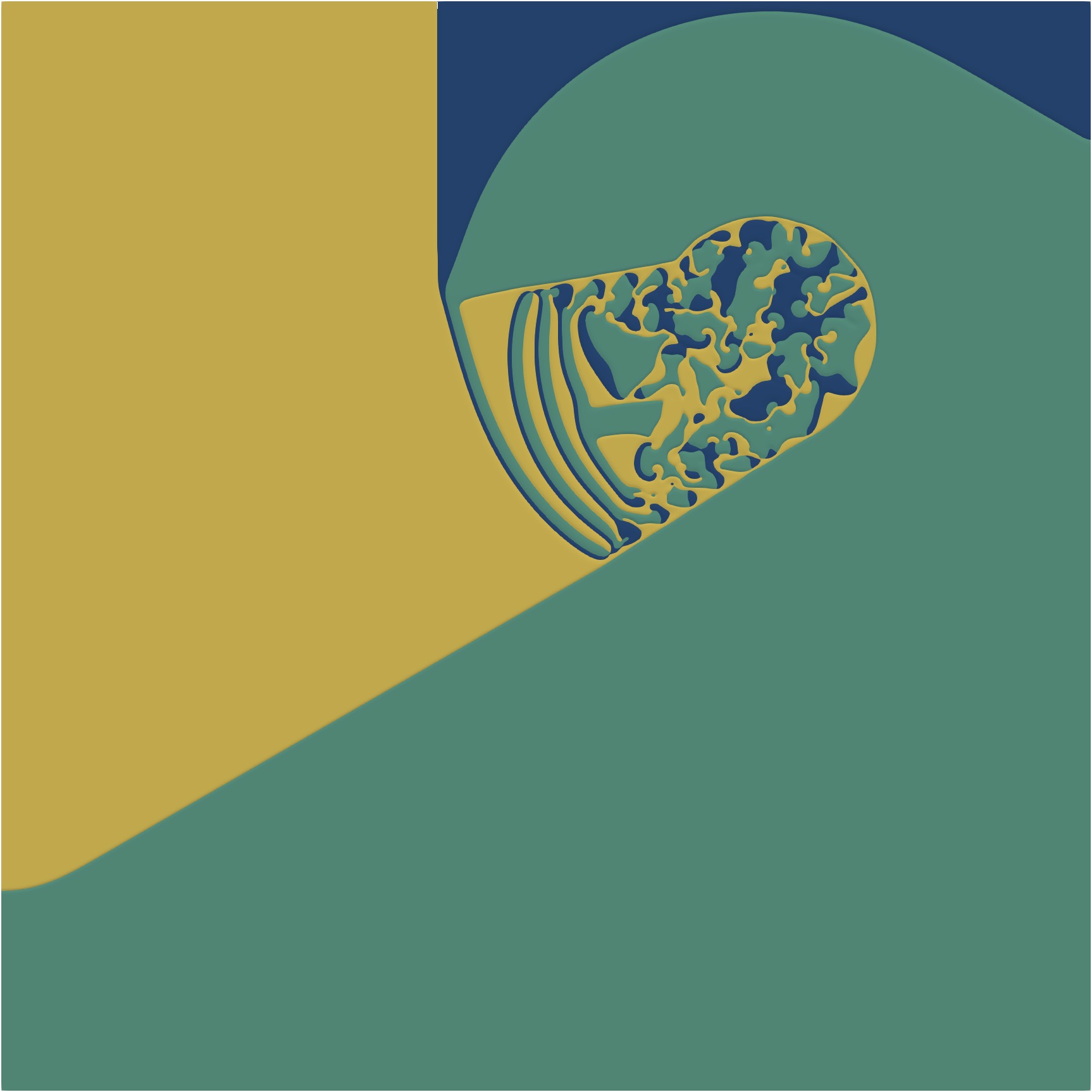}
    }\hfill
    \\[0em]
    \hfill\subcaptionbox{$t = 960$}[0.495\textwidth]{%
        \includegraphics[width=0.4\textwidth]{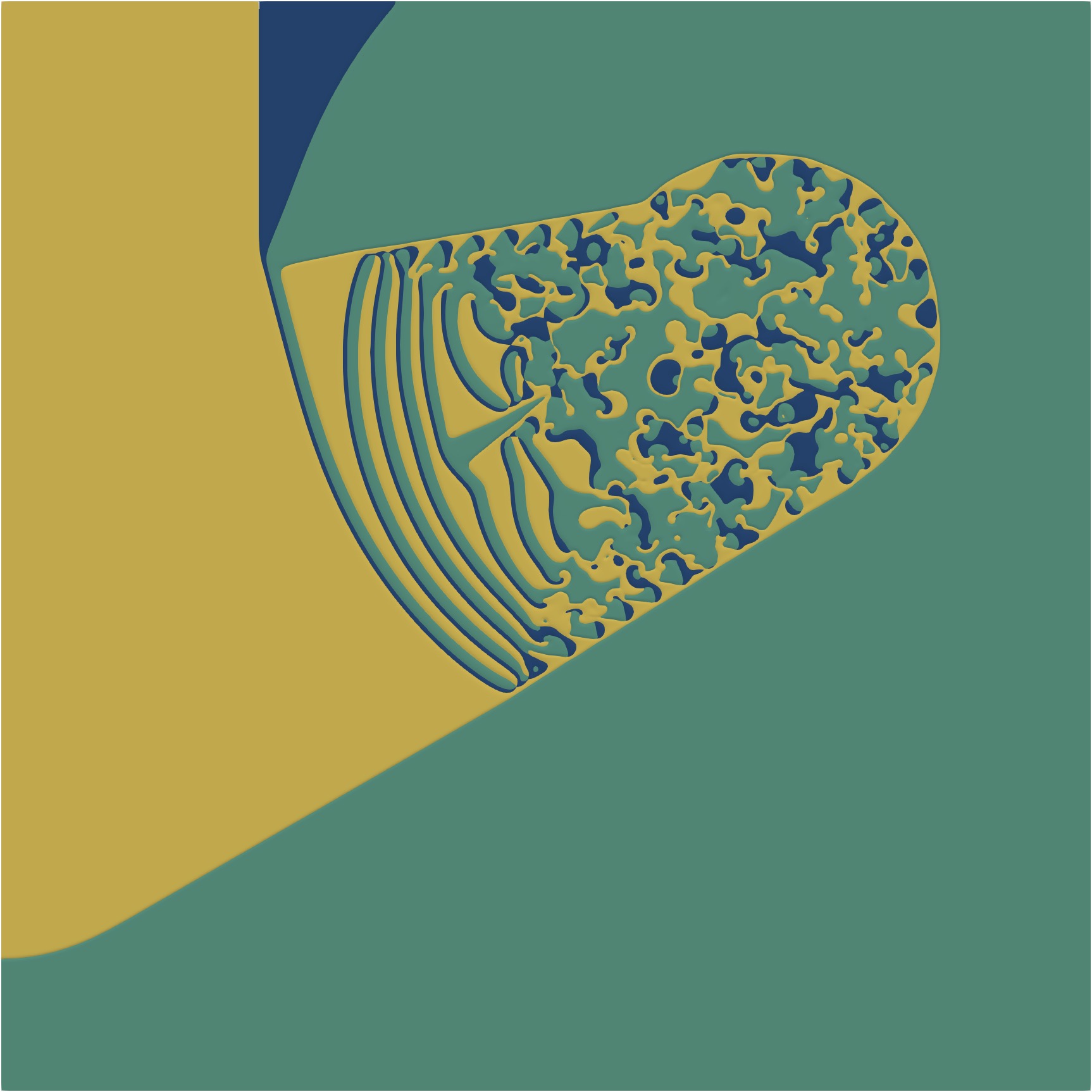}
    }\hfil%
    \subcaptionbox{$t = 1,400$}[0.495\textwidth]{%
        \includegraphics[width=0.4\textwidth]{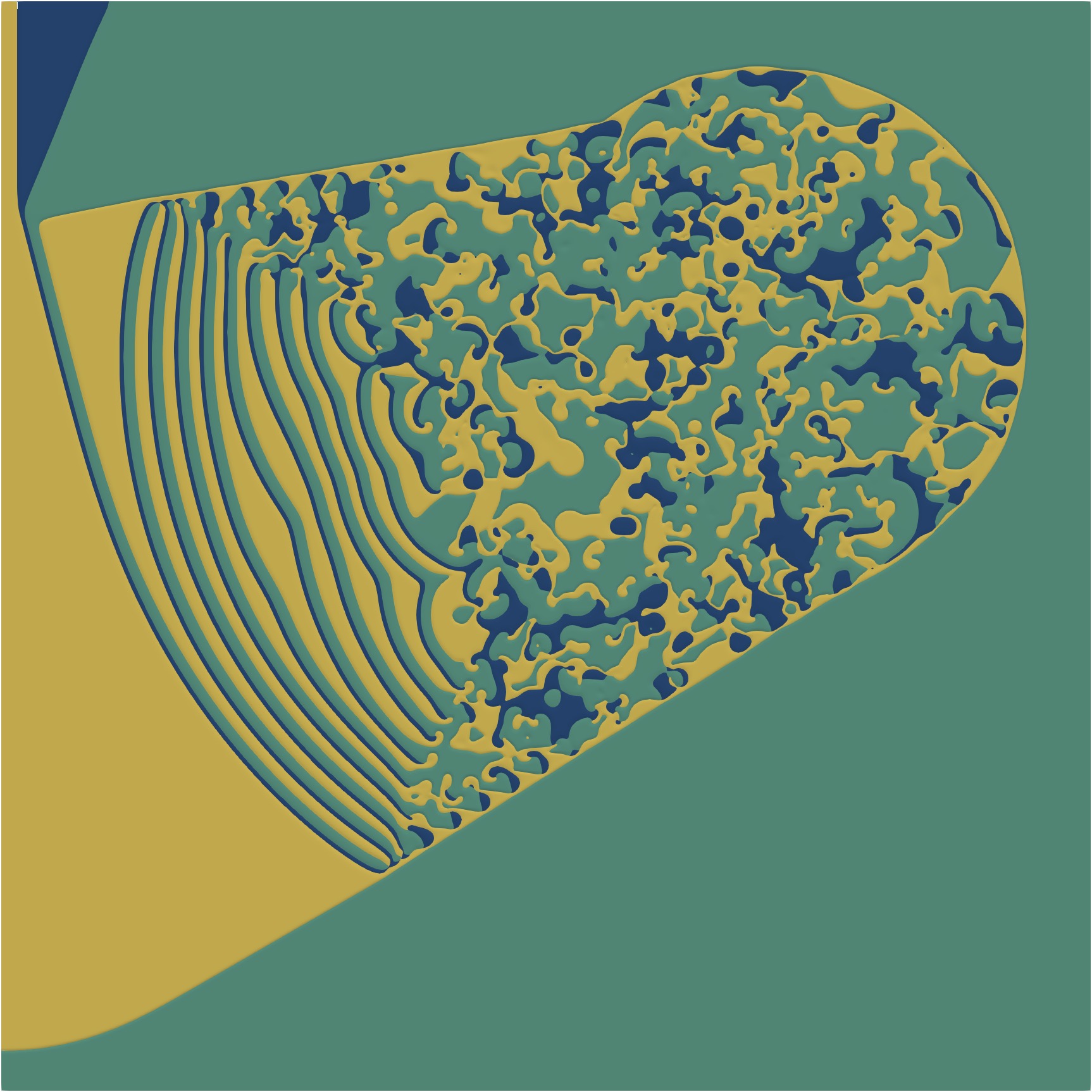}
    }\hfill
    \caption{The evolution of the `band' patterns which are observed when $\diff = (1, 0.1, 0.9)^{\top}$ and same the same initial conditions as in Fig.~\ref{fig:cycliccompetition:icecream}.}
\label{fig:stripes}
\end{figure}

The transition from the pattern of spread via droplet-like units (Fig.~\ref{fig:cycliccompetition:icecream}) to the one containing bands (Fig.~\ref{fig:stripes}) can be understood by exploring the schematic diagram shown in Fig.~\ref{fig:schematic}.
The spread of the droplets in the wedge can be described by considering pairwise interactions of species, most of which actually occur via plane wave interactions. We can neglect the presence of a third species since the density of each species rapidly drops when entering the domain dominated by another species (except the points where all three species meet, as at the tip of the wedge). In Fig.~\ref{fig:schematic}, we show the direction of the spread of plane waves of cyclic displacement of species; here $C_{i,j}$ denotes the speed of the plane wave replacing species $j$ by its stronger competitor $i$. One can also see a round interface between species 1 and species 2. The corresponding wave speed is denoted by $V_{1,2}=V_{1,2} (R)$, where $R$ is the radius of the curvature. The values of  $C_{i,j}$ and  $V_{1,2}$ can be determined by considering the one-dimensional case (in the case of $V_{1,2}$  one should explore the system in polar coordinates). Our simulations show that for the parameters from Fig.~\ref{fig:cycliccompetition:icecream}, in the one-dimensional case the prorogation of a travelling pulse composed of all three species is impossible, whereas for pairwise switch waves we have $C_{1,2}>C_{2,3}$. The curvature of the wave reduces the spread of the propagation of the front of species 1 in the droplet, thus $C_{1,2}>V_{1,2} (R)= C_{2,3}$ and the spread of the droplet becomes synchronised with $C_{2,3}$ and this gives the condition for $R$. However, in the case where the diffusion of species 3 increases (as in Fig.~\ref{fig:stripes}), our one-dimensional simulations show that $C_{2,3}>C_{1,2}$, thus $C_{2,3}>V_{1,2}(R)$ for any $R$ and the propagation of a travelling pulse composed of all three species now becomes possible. As a result, the droplet-shape structure breaks down and a one-dimensional band composed of three species is eventually formed. In the case of Fig.~\ref{fig:stripes}, our simulations demonstrate that the propagation of one-dimensional pulses is stable with respect to the two-dimensional case, i.e. small two-dimensional perturbations would not destabilise them; however, this does not provide an explanation of why formation of bands does not occur in opposite direction (i.e. towards the right), thus more investigation is needed to fully understand this problem.

\begin{figure}
\includegraphics[width=\textwidth]{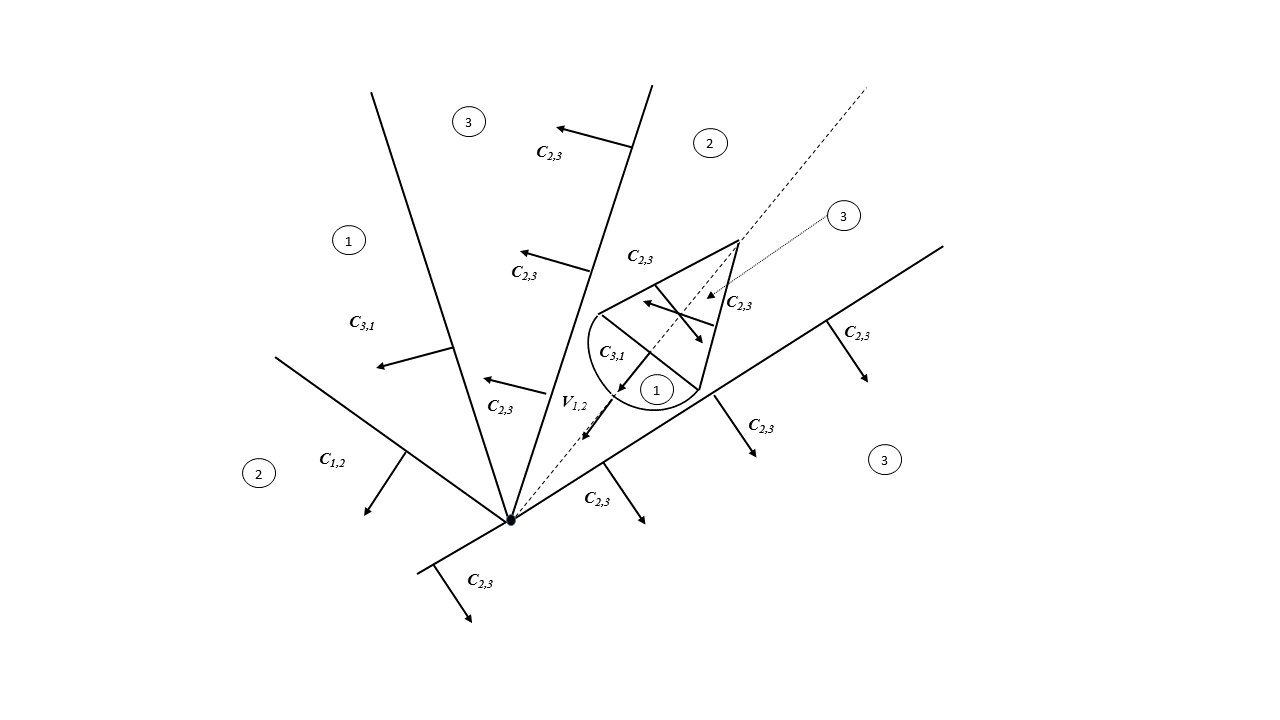}\caption{Schematic representation of the movement of a droplet-like unit. A detailed description is given in the main text.}
\label{fig:schematic}
\end{figure}

From Fig.~\ref{fig:cycliccompetition:icecream} one can also observe another new type of transient pattern which we call the dynamical droplets. Such structures are formed along the boundary of the domain in which the species coexist. This is essentially a transient regime, however, since although their life time can be long, the dynamical droplets structures will eventually collide and annihilate one another, resulting in chaotic dynamics. The tips of dynamical droplet pattern is a point where high densities of all three species meet each other. One can see in the figure that the growing dynamical droplet structure is generated by a certain small area which can be considered as a generator. This is similar to a target wave emanating from a pacemaker; however, the generation of dynamical droplet-shaped moving patches is a transient phenomenon that eventually stops after collision with irregular patches.

We also considered the other scenario of cyclic competition: conditional cyclic competition. In particular it is realised for  $\alpha_{3,1} = 1.3$ (the other interaction parameters being the same as before) for which local interactions between $\uvec[3]$ and $\uvec[1]$ are bistable. Our numerical study revealed a new pattern of invasion domain shown in Fig.~\ref{fig:cycliccompetition:glider}.
The figure corresponds to $\epsilon_{2}=0.55$ and $\epsilon_{3}=0.5$. The initial spiral tip emanates a droplet composed of species 2 and 3. The droplet takes the form of a the glider which moves towards the lower left corner. Some glider-shaped structures quickly disappear whereas some can persist by splitting and merging with other patches. The tip which generates glider-shaped structures eventually reaches the boundary and disappear. Irregular patches produced in this dynamical regime can persist for a long time; however, eventually they disappear and only one species survives. Thus, this new regime of invasion can only guarantee species coexistence at the moving front of invasion. Our investigation of the parameter space of diffusion coefficients reveals that this pattern is robust and is observed within a 10\% variation of $\epsilon_{2}=0.55$ and $\epsilon_{3}=0.5$.

\begin{figure}
    \subcaptionbox{$t = 60$}[0.495\textwidth]{%
        \includegraphics[width=0.49\textwidth]{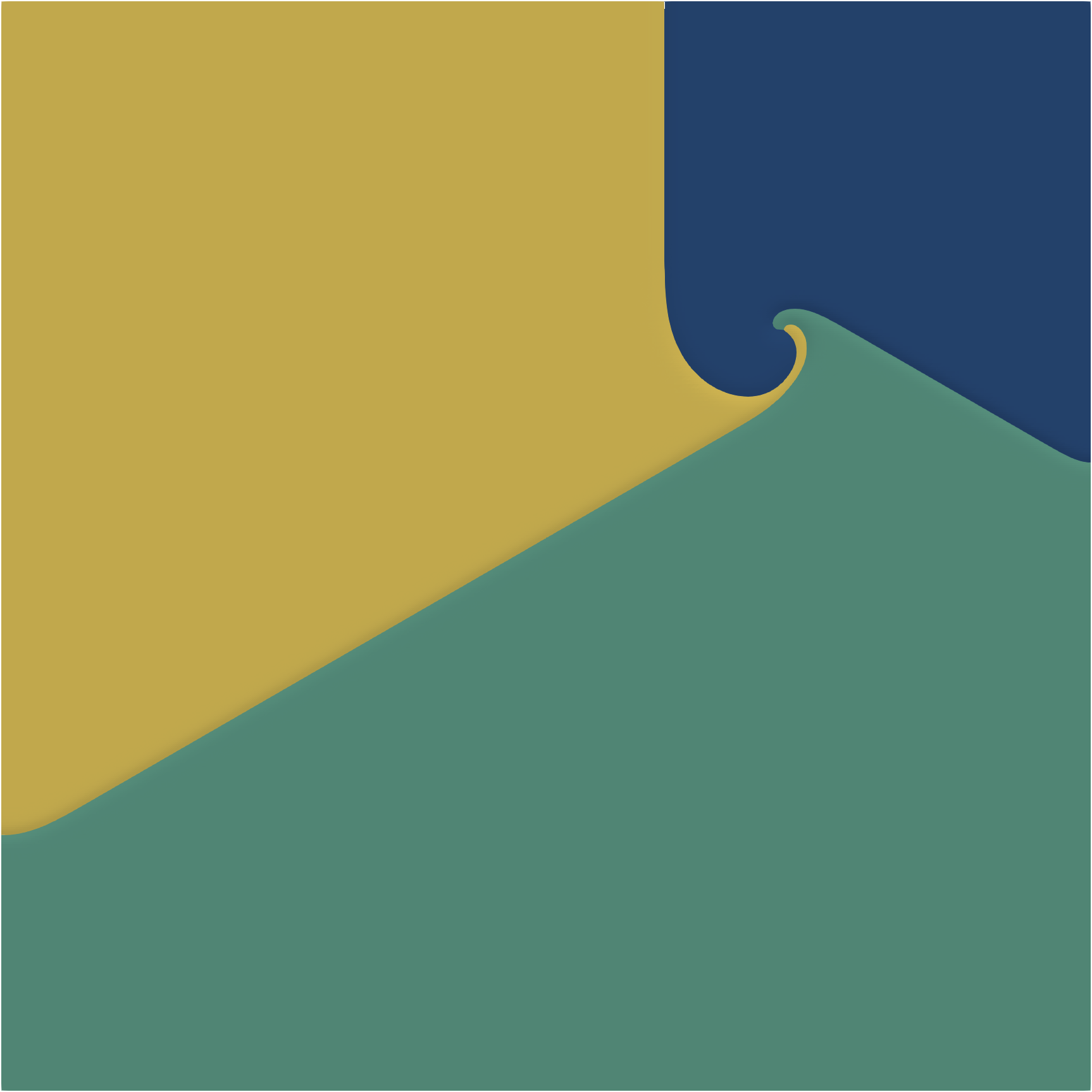}
    }\hfill%
    \subcaptionbox{$t = 120$}[0.495\textwidth]{%
        \includegraphics[width=0.49\textwidth]{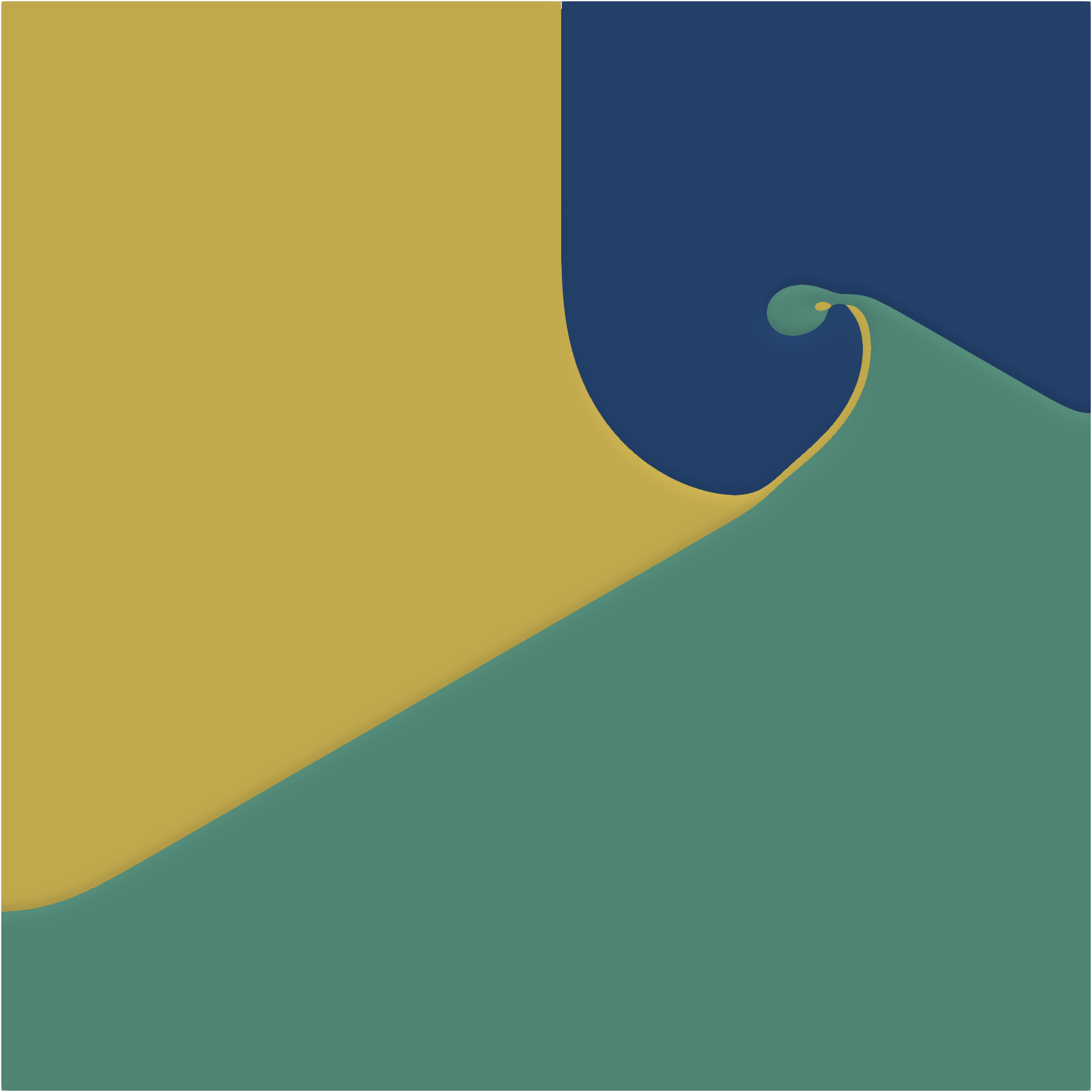}
    }
    \\[1em]
    \subcaptionbox{$t = 160$}[0.495\textwidth]{%
        \includegraphics[width=0.49\textwidth]{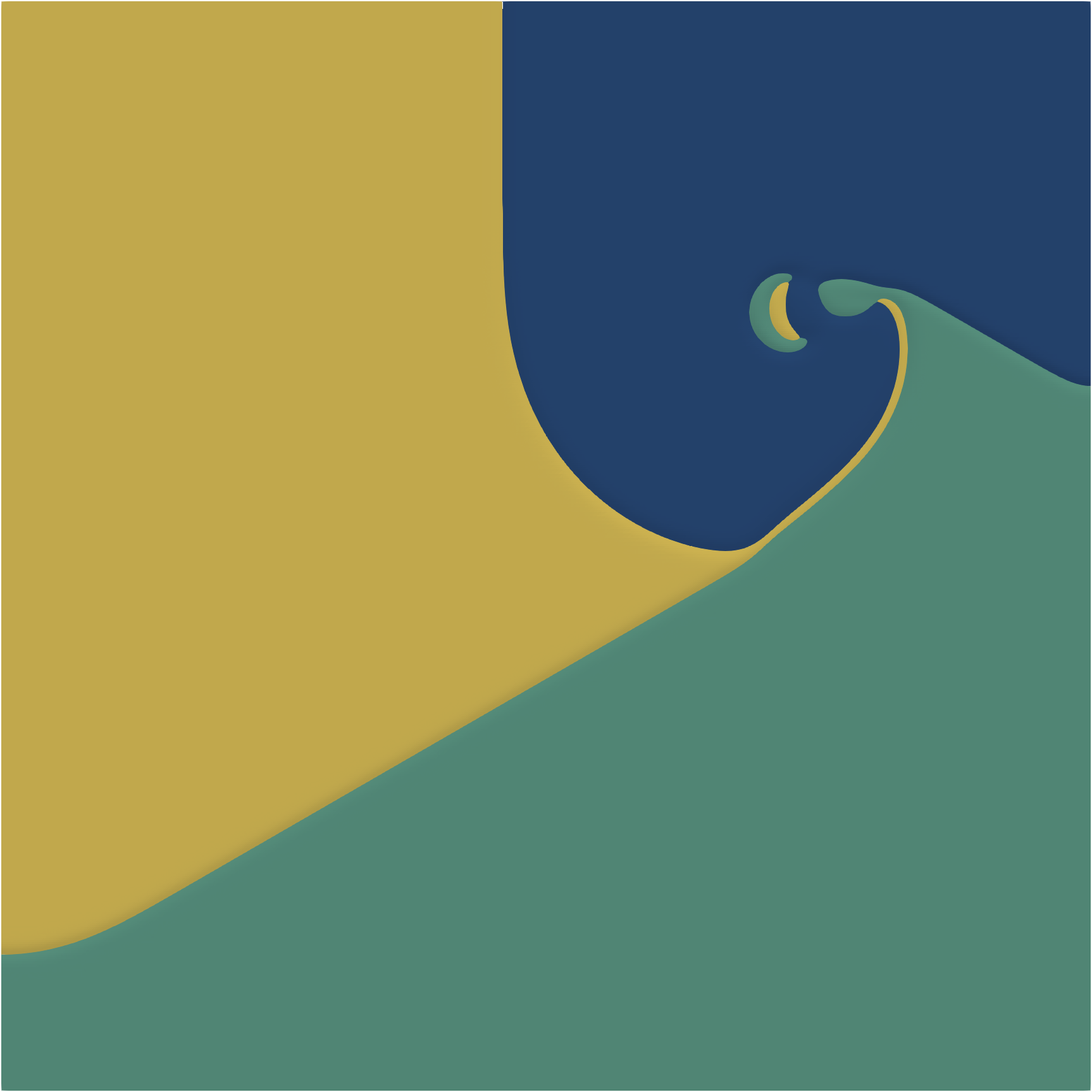}
    }\hfill%
    \subcaptionbox{$t = 290$\label{fig:glider_290}}[0.495\textwidth]{%
        \includegraphics[width=0.49\textwidth]{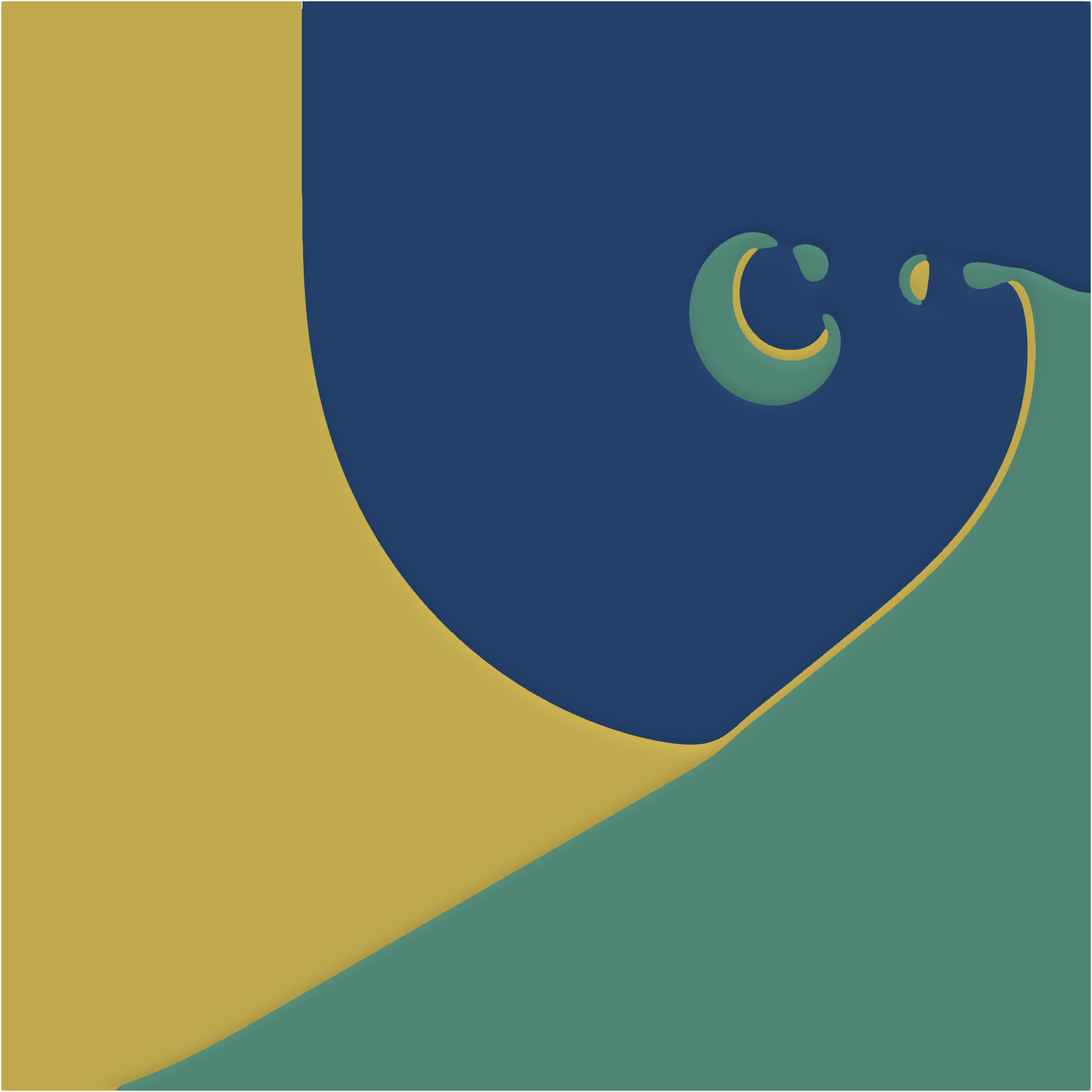}
    }
    \caption{The evolution of the `glider' patterns which are observed when $\epsilon_{2}=0.55$ and $\epsilon_{3}=0.5$, $\alpha_{3,1} = 1.3$ and the initial conditions are the same as in Fig.~\ref{fig:cycliccompetition:icecream}. The three colours indicate the regions of the domain in which each of the three competing species dominates.}
    \label{fig:cycliccompetition:glider}
\end{figure}

Finally, we extend our analysis to the three-dimensional case. We focus on exploring the three-dimensional analogue of the regular droplet-like structured observed in the two-dimensional model under classical cyclic competition shown in Fig.~\ref{fig:cycliccompetition:icecream}. Our results when the species are initially distributed within separate cubic regions are shown in Fig.~\ref{fig:cycliccompetition:3D}.

\begin{figure}
    \hfill\subcaptionbox*{}[0.05\textwidth]{%
        \begin{sideways}$\qquad\qquad\qquad\qquad t=0$\end{sideways}
    }\hfil%
    \subcaptionbox*{}[0.47\textwidth]{%
        \includegraphics[width=0.4\textwidth]{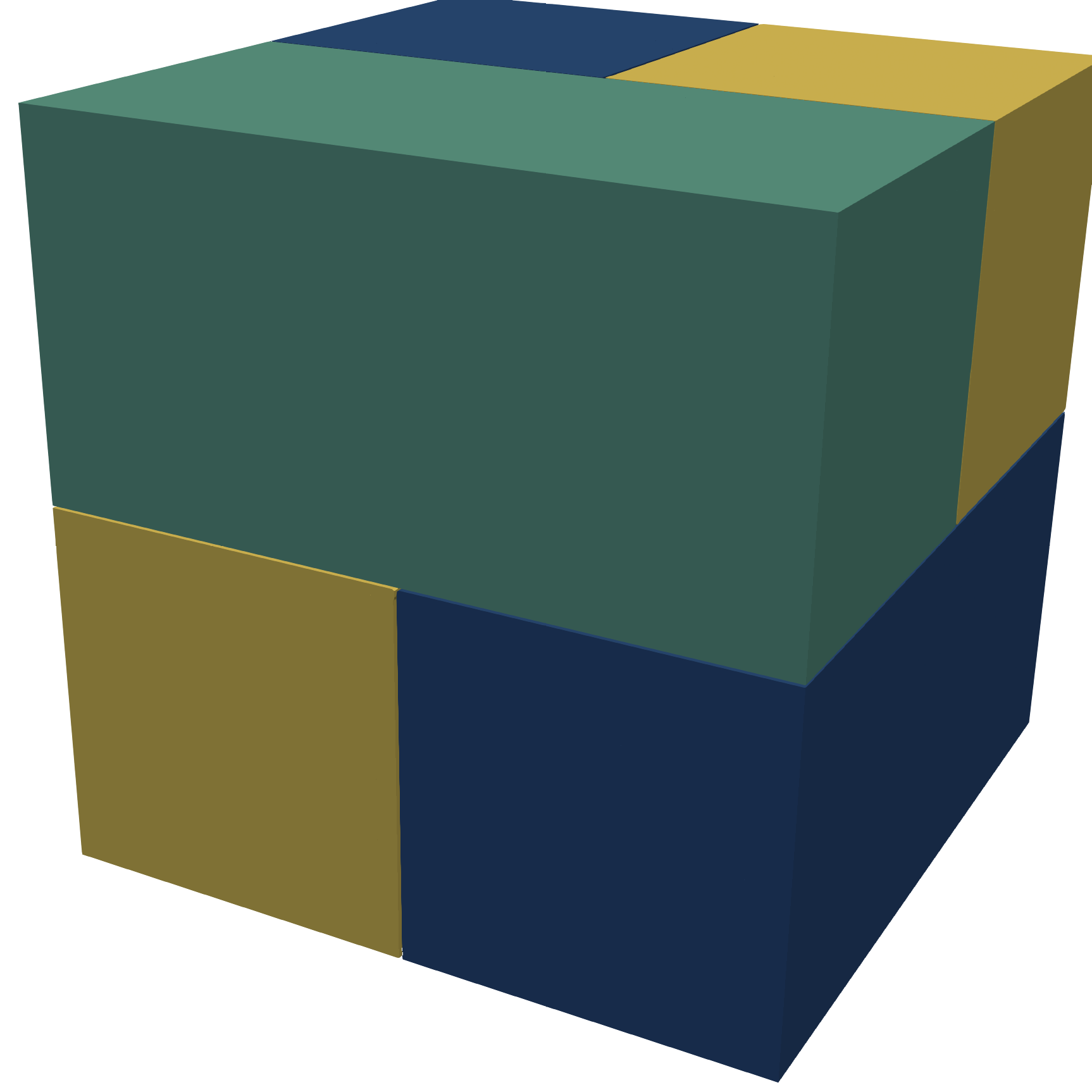}
    }\hfil%
    \subcaptionbox*{}[0.47\textwidth]{%
        \includegraphics[width=0.4\textwidth]{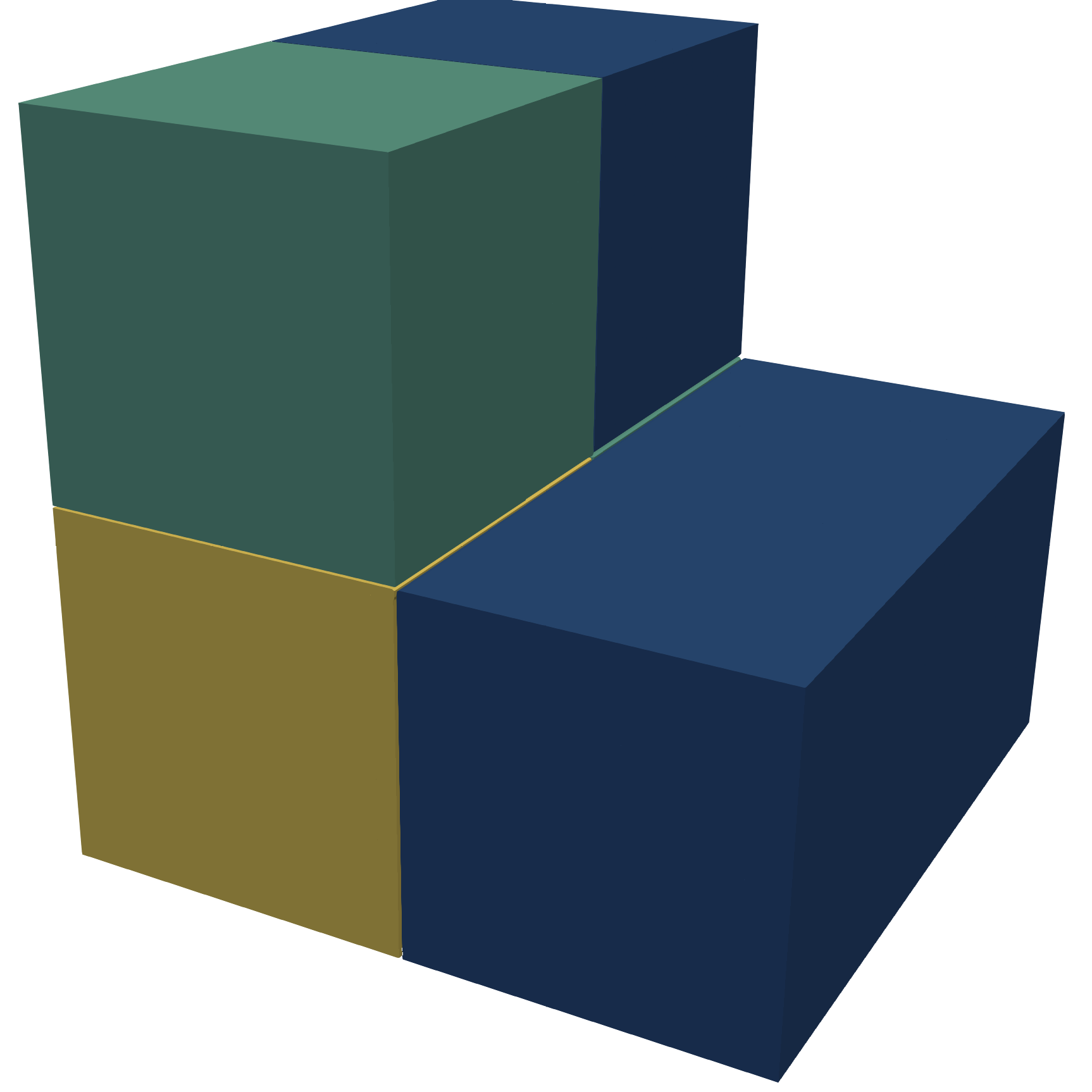}
    }\hfill
    \\[-0.6em]
    \hfill\subcaptionbox*{}[0.05\textwidth]{%
        \begin{sideways}$\qquad\qquad\qquad\qquad t=160$\end{sideways}
    }\hfil%
   \subcaptionbox*{}[0.47\textwidth]{%
        \includegraphics[width=0.4\textwidth]{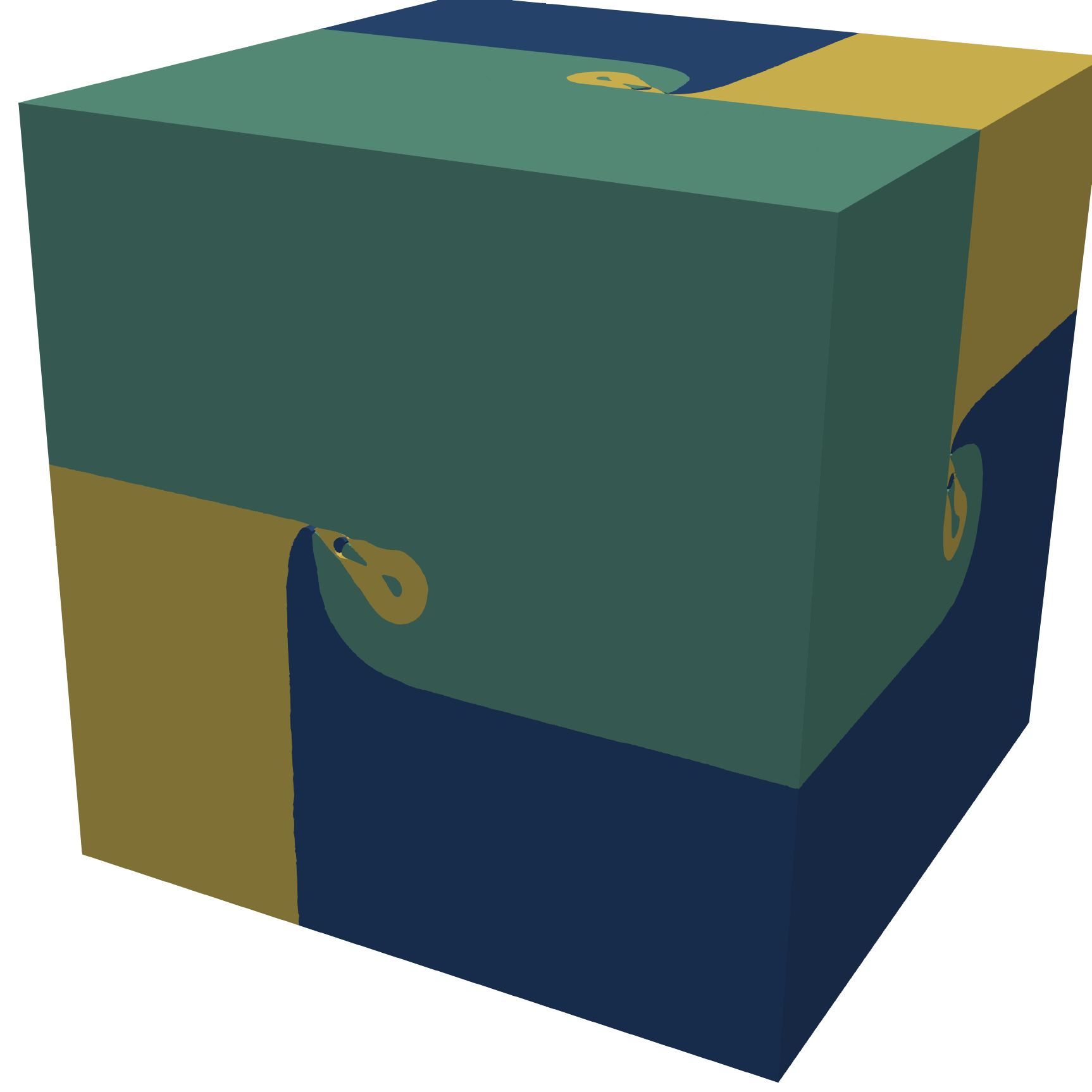}
    }\hfil%
    \subcaptionbox*{}[0.47\textwidth]{%
        \includegraphics[width=0.4\textwidth]{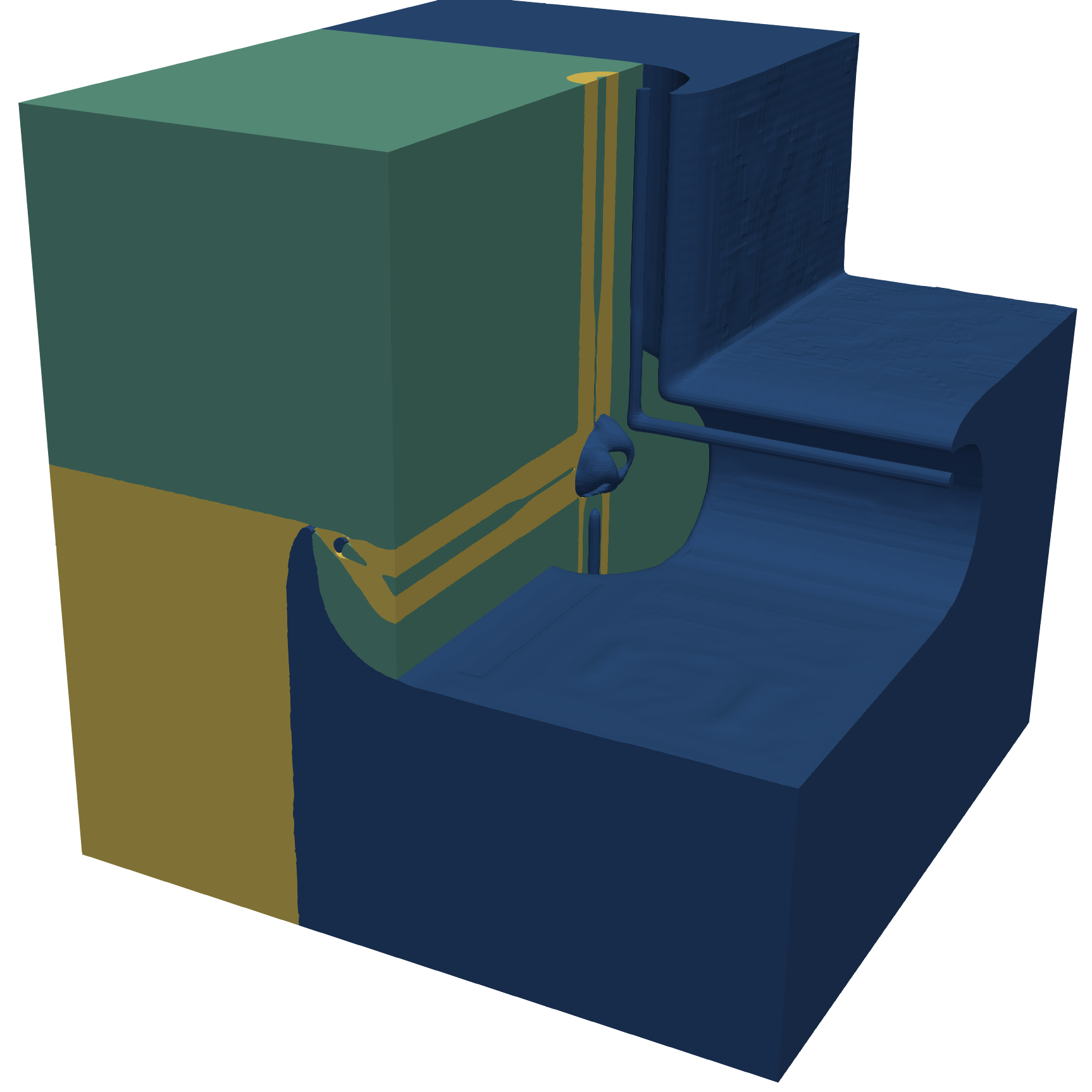}
    }\hfill
    \\[-0.6em]
    \hfill\subcaptionbox*{}[0.05\textwidth]{%
        \begin{sideways}$\qquad\qquad\qquad\qquad t=560$\end{sideways}
    }\hfil%
    \subcaptionbox*{}[0.47\textwidth]{%
        \includegraphics[width=0.4\textwidth]{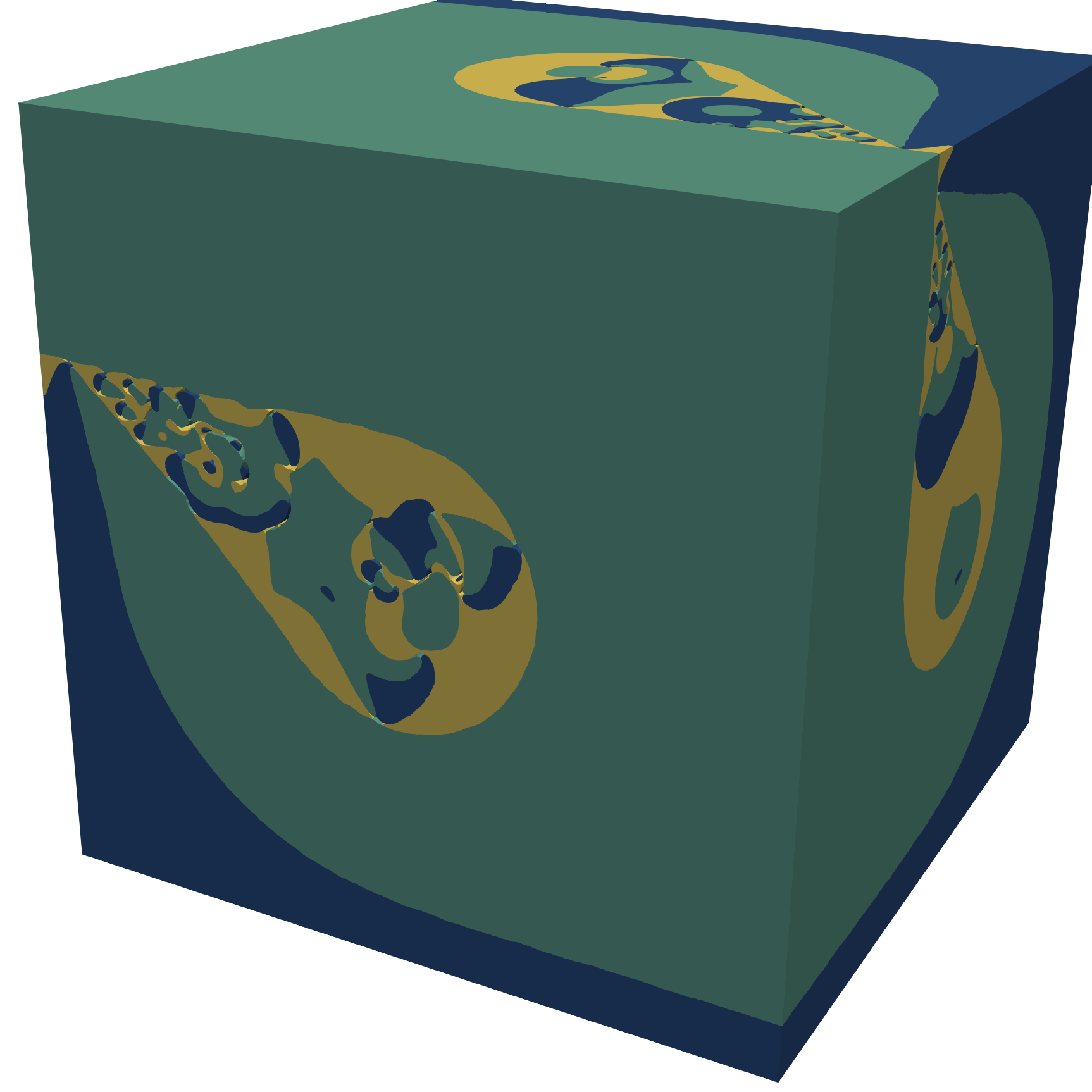}
    }\hfil%
    \subcaptionbox*{}[0.47\textwidth]{%
        \includegraphics[width=0.4\textwidth]{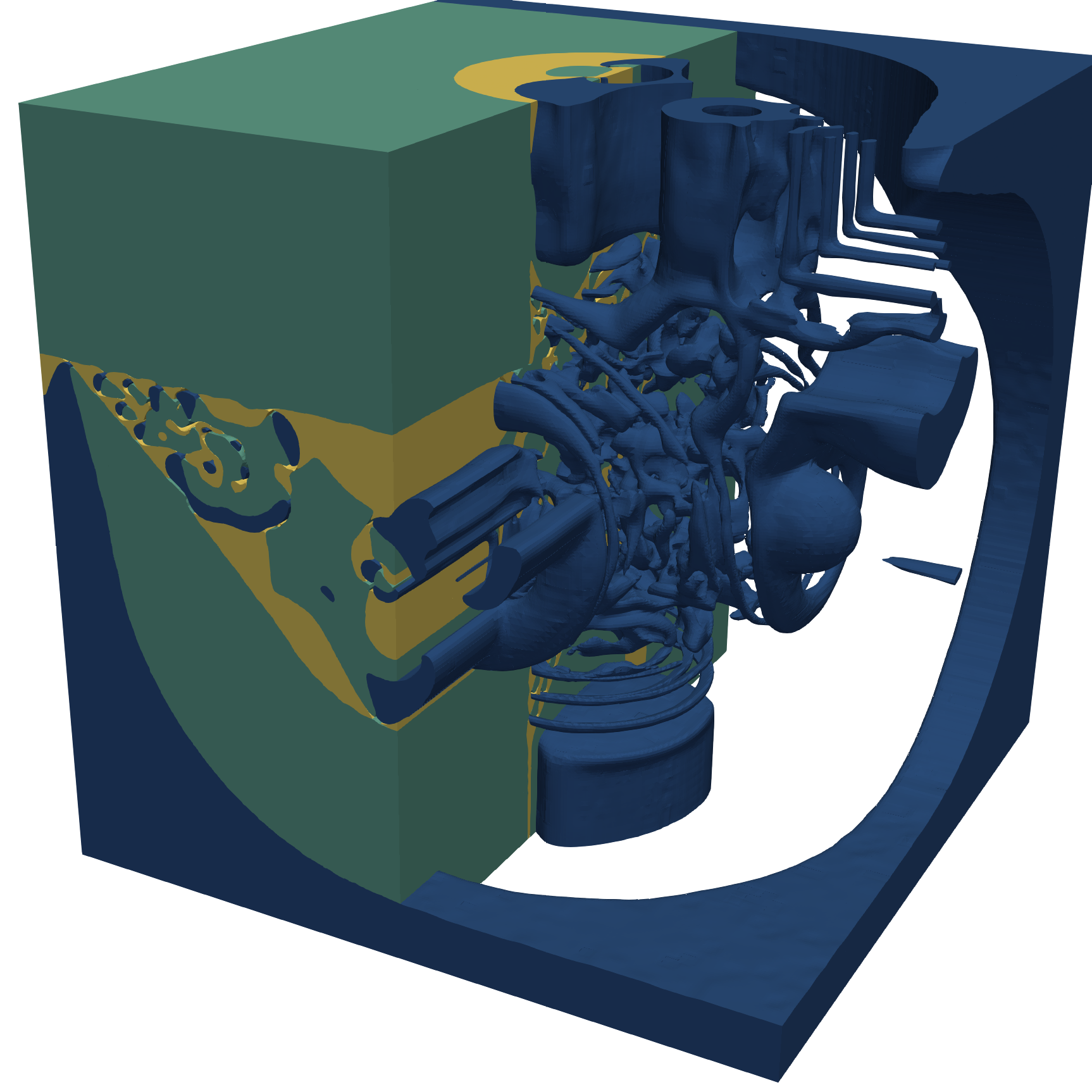}
    }\hfill
    \\[-1em]
    \caption{Evolution in three spatial dimensions using the parameters $\epsilon_{2}=0.1$ and $\epsilon_{3}=0.6$ (as in Fig.~\ref{fig:cycliccompetition:icecream}), computed on the domain $\Omega = [0,600]^3$. The figures in the left column show the domains dominated by each species, while those on the right show  $\uvec[2]$ and $\uvec[3]$ only in the subdomain $x>300$.}
    \label{fig:cycliccompetition:3D}
\end{figure}

We can observe droplet-like regular structures as well. However, such patterns appear to be prismatic analogues of the 2-dimensional structures. We found as well that the long time dynamics are essentially irregular and chaotic: the area occupied by irregular patches of species coexistence will eventually invade the entire domain $\Omega$ (the result is not shown here for brevity). Note that similar results are obtained when the species initially dominate within overlapping yet mutually exclusive spherical regions, although for brevity these are not reported here.

\section{A posteriori error analysis}\label{sec:apost}
The adaptive numerical scheme described in Section~\ref{sec:numericalscheme}, and used to compute the numerical results presented here, relies on a \emph{local error indicator functional} which is used to decide which mesh elements are refined or coarsened.
This error indicator is derived from a fully computable estimate for the numerical error measured in the $L^2(0,t; H^1(\Omega))$ norm in the form of a residual-type a posteriori error estimate, developed in this section.

We develop the a posteriori error analysis  for the spatially semidiscrete method~\eqref{eq:finiteElementForm}; this choice is commented upon at the end of the section. To keep the presentation clear, the analysis we present here focusses on the single-equation situation (i.e. when $m=1$), although all the results follow for when $m>1$ in a completely analogous fashion.
In this case, $\diff > 0$ is simply a constant scalar.

We introduce the function $\fh : H^1(\domain) \to \Vh$ such that, for any $v \in H^1(\domain)$,
\begin{equation*}
	(\fh(v), \testfn) = (f(v), \testfn) \qquad \forall \testfn \in \Vh,
\end{equation*}
so that $\fh(v)$ is the $L^2(\domain)$-orthogonal projection of $f(v)$ into $\Vh$.
Next, we define the discrete Laplacian operator $\Delta_h:\Vh \to \Vh$ to be such that, for any $\vh \in \Vh$,
\begin{align*}
	(- \diff \Delta_h \vh, \testfn) = (\diff \nabla \vh, \nabla \testfn) \qquad \forall \testfn \in \Vh.
\end{align*}
Following~\cite{MN03}, we also define the elliptic reconstruction operator $\elliprecon{}: \Vh \to H^1(\domain)$ to satisfy, for any $\vh \in \Vh$,
\begin{equation*}
	(\diff \nabla \elliprecon{\vh}, \nabla v) = (g_h(\vh), v) \qquad \forall v \in H^1(\domain),
\end{equation*}
where $g_h(\vh) := - \diff \Delta_h \vh - \fh(\vh) + f(\vh)$.
Since $\uh$ thus coincides with the finite element solution to the elliptic problem with true solution $\elliprecon{\uh}$, we assume that we have at our disposal a computable a posteriori error bound of the form
\begin{equation}\label{eq:associatedEllipticBound}
	\norm{\elliprecon{\uh} - \uh}_{H^{\reg}} \leq \Celip E_\reg(\uh)
\end{equation}
for $\reg = 0,1$, where $H^0 = L^2$; see, e.g., \cite{AO00,V13} for several examples of such results.
For instance, a residual-type a posteriori bound on the error would be of the form
\begin{equation}\label{eq:ellipticEstimator}
	E_r(\uh)^2 = \sum_{\element \in \Th} h_{\element}^{2{\reg}} \norm{\diff \Delta \uh - \diff \Delta_h \uh - \fh(\uh) + f(\uh)}^2_{L^2(\element)} + \sum_{\edge \in \edges} h_{\edge}^{2\reg - 1} \norm{ \diff \jump{ \nabla \uh} }^2_{L^2(\edge)},
\end{equation}
where $\Th$ denotes the set of mesh elements, $\edges$ denotes the set of mesh faces, and $h_{\element}$ and $h_{\edge}$ denote the diameter of a mesh element and face, respectively.
The notation $\jump{\nabla \uh}$ is used to denote the \emph{jump} of $\nabla \uh$ over a mesh face, with the definition that, for a point $\vx \in \edge$, $\jump{\nabla \uh}(\vx) = \vec{n} \cdot ((\nabla \uh)^+(\vx) - (\nabla \uh)^-(\vx))$, where $(\nabla \uh)^+(\vx)$ and $(\nabla \uh)^-(\vx)$ are the values at $\vx$ of $\nabla \uh$ on each of the elements $\element^+$ and $\element^-$ meeting at the face $\edge$ respectively, and $\vec{n}$ denotes the unit vector normal to $\edge$ pointing from $\element^+$ to $\element^-$.
In the interest of brevity, throughout the remainder of this section, we shall use $\norm{\cdot}$ to denote the norm on $L^2(\domain)$.
Within the analysis, we shall make use of the following  consequence of H\"{o}lder's inequality
\begin{equation}
	\int_\domain \alpha^{\gamma+1}\beta \dx \leq \frac{\gamma+1}{\gamma+2} \int_\domain \alpha^{\gamma+2} \dx + \frac{1}{\gamma+2} \int_\domain \beta^{\gamma+2} \dx,
	\label{eq:integralProductInequality}
\end{equation}
that is valid for $\alpha, \beta > 0$, and the result of Lemma 5.1 in~\cite{CGJ13}, that for any $v \in H^1(\domain)$
\begin{equation}
	\norm{v}_{L^{\gamma+2}(\domain)}^{\gamma+2} \leq \Cdom \norm{v}^\gamma \norm{\nabla v}^2,
	\label{eq:gammaPlus2Inequality}
\end{equation}
where $\Cdom$ is a constant depending only on the domain $\domain$. We are now in a position to derive an a posteriori bound on the error between the exact and the finite element solution.

\begin{theorem}\label{aposteriori}
Let $\u$ be the solution of~\eqref{eq:weakForm} and $\uh$ the solution of~\eqref{eq:finiteElementForm}, respectively,
under the assumption that the nonlinear term $f$ satisfies the Lipschitz-style growth condition~\eqref{eq:growthBound}.
The error  $\error := \u - \uh$ satisfies the $L^{\infty}(0,t; L^2(\Omega))$ error estimate
\begin{align*}
	C\max_{s \in [0,t]}\norm{\error(s)}^2 &\leq \max_{s \in [0,t]} E_0(\uh(s))
	+ e^{\expconst t} \int_0^t E_0(\uht(s))^2 + 4M \Cf E_0(\uh(s))^2 \ds
	\\&\quad
	+ e^{\expconst t} \int_0^t C_1 E_0(\uh(s))^\gamma E_1(\uh(s))^2 \ds,
\end{align*}
and the $L^2(0,t; H^1(\Omega))$ error estimate
\begin{align*}
	C\int_0^t \diff \norm{ \nabla \error}^2 \d s &\leq \int_0^t \diff E_1(\uh(s))^2 \ds
	+ e^{\expconst t} \int_0^t E_0(\uht(s))^2 + 4M \Cf E_0(\uh(s))^2 \ds
	\\&\quad
	+ e^{\expconst t} \int_0^t C_1 E_0(\uh(s))^\gamma E_1(\uh(s))^2 \ds,
\end{align*}
where $M := \sup_{t \in [0,T]}(1 + 2\norm{\uh(t)}_{L^\infty(\domain)}^\gamma)$.
In each case, the positive constant $C$ is independent of $h$, $\u$ and $\uh$.
\end{theorem}
\begin{proof}
Using the definitions of $\fh$ and the discrete Laplacian operator, we can rewrite~\eqref{eq:finiteElementForm} as
\begin{equation*}
	(\uht - \diff \Delta_h \uh - \fh(\uh), \testfn) = 0 \qquad \forall \testfn \in \Vh,
\end{equation*}
implying that $\uh$ satisfies
\begin{equation*}
	\uht - \diff \Delta_h \uh - \fh(\uh) = 0,
\end{equation*}
and therefore, from the definition of the elliptic reconstruction,
\begin{equation}\label{eq:withEllipticProjection}
	(\uht, v) + (\diff \nabla \elliprecon{\uh}, \nabla v) + (\f(\uh), v) = 0	\qquad \forall v \in H^1(\domain).
\end{equation}

To derive the required bounds, we first  decompose  the error $\error = \u - \uh$ into $\rho := \u - \elliprecon{\uh}$ and $\theta := \elliprecon{\uh} - \uh$.
Since~\eqref{eq:associatedEllipticBound} provides a bound on the reconstruction error $\theta$, we focus principally on deriving a bound for $\rho$.
Subtracting~\eqref{eq:withEllipticProjection} from the original weak form~\eqref{eq:weakForm}, we find that
\begin{equation*}
	(\rho_t, v) + (\diff \nabla \rho, \nabla v) = (-\theta_t, v) + (\f(\u) - \f(\uh), v) \qquad \forall v \in H^1(\domain),
\end{equation*}
and therefore, picking $v = \rho \in H^1(\domain)$,
\begin{equation}\label{eq:beforeNonlinearBounding}
	\frac{1}{2} \frac{d}{dt} \left( \norm{\rho}^2 \right) + \diff \norm{\nabla \rho}^2 \leq \norm{\theta_t}^2 + \norm{\rho}^2 + (\f(\u) - \f(\uh), \rho).
\end{equation}

To treat the nonlinear term in~\eqref{eq:beforeNonlinearBounding}, we use assumption~\eqref{eq:growthBound} on the growth of $\f$, to find
\begin{align*}
	\abs{(\f(\u) - \f(\uh), \rho)}
		&\leq \int_\domain \abs{\f(\u) - \f(\uh)}\abs{\rho} \dx
		\leq \Cf \int_\domain \left( 1 + \abs{\u}^\gamma + \abs{\uh}^\gamma \right) \abs{\u - \uh} \abs{\rho} \dx.
\end{align*}
The restriction on $\gamma$ implies that $\abs{a+b}^\gamma \leq 2^{\max\{1, \gamma\} - 1} (\abs{a}^\gamma + \abs{b}^\gamma) \leq  2 (\abs{a}^\gamma + \abs{b}^\gamma)$, from which, with $a = \u-\uh$ and $b = \uh$, it follows that
\begin{equation}\label{eq:twoTermBound}
	\abs{(f(\u) - f(\uh), \rho)} \leq 2 \Cf \int_\domain M \abs{\u-\uh} \abs{\rho} + \abs{\u-\uh}^{\gamma+1} \abs{\rho} \dx.
\end{equation}

The first term on the right hand side of~\eqref{eq:twoTermBound} can be bounded as
\begin{align}\label{eq:firstTermSimplification}
	\int_\domain \abs{\u-\uh} \abs{\rho} \dx &\leq \int_\domain \abs{\rho}^2 + \abs{\theta \rho} \dx \leq  \frac{3}{2} \norm{\rho}^2 + \frac{1}{2}\norm{\theta}^2 .
\end{align}
For the second term on the right hand side of~\eqref{eq:twoTermBound}, however, we make use of the bounds~\eqref{eq:integralProductInequality} and~\eqref{eq:gammaPlus2Inequality}, yielding
\begin{align*}
	\Cf\int_\domain \abs{\u-\uh}^{\gamma+1} \abs{\rho} \dx &\leq \Cf\frac{\gamma+1}{\gamma+2} \int_\domain \abs{\u-\uh}^{\gamma+2} \dx + \frac{\Cf}{\gamma+2} \int_\domain \abs{\rho}^{\gamma+2} \dx	\\
		&\leq \Cf \frac{(\gamma+1)2^{\gamma + 1}}{\gamma+2} \|\theta\|_{L^{\gamma+2}(\domain)}^{\gamma+2} + \Cf \frac{1+ (\gamma+1)2^{\gamma + 1}}{\gamma+2} \|\rho\|_{L^{\gamma+2}(\domain)}^{\gamma+2} \\
        &\leq C_1 \|\theta\|^\gamma \|\nabla \theta\|^2 + C_2 \|\rho\|^\gamma \|\nabla \rho\|^2,
\end{align*}
where
\begin{align*}
    C_1 := \Cf \Cdom 2^{\gamma + 1} \frac{\gamma+1}{\gamma+2} \quad \text{ and } \quad C_2 := C_1 + \Cf \Cdom(\gamma + 2).
\end{align*}
Combining this with~\eqref{eq:firstTermSimplification}, the error bound~\eqref{eq:beforeNonlinearBounding} becomes
\begin{equation*}
	\frac{1}{2} \frac{d}{dt} \left( \norm{\rho}^2 \right) + \diff \norm{ \nabla \rho}^2 \leq \norm{\theta_t}^2 + 4M \Cf \norm{\theta}^2 + \expconst \norm{\rho}^2 + C_1 \norm{\theta}^\gamma \norm{\nabla \theta}^2 + C_2 \norm{\rho}^\gamma \norm{\nabla \rho}^2.
\end{equation*}
Integrating through time, we find that
\begin{equation*}
	\norm{\rho(t)}^2 + \int_0^t \diff \norm{ \nabla \rho}^2 \dt  \leq \delta(\theta)^2 + \expconst \int_0^t \norm{\rho}^2 \dt + C_2 \int_0^t \norm{\rho}^\gamma \norm{\nabla \rho}^2 \dt,
\end{equation*}
where we observe that, by construction, $\rho(0) = 0$, and the functional $\delta$ is given by
\begin{equation*}
	\delta(\theta)^2 := \int_0^T\|\theta_t\|^2 + 4M \Cf \|\theta\|^2 + C_1 \|\theta\|^\gamma \|\nabla \theta\|^2 \dt.
\end{equation*}
The inequality $a^\gamma b \leq \left( a^2 + b \right)^{1 + \frac{\gamma}{2}}$ (a consequence of Young's inequality) then implies that
\begin{align}\label{eq:beforeHiding}
	\norm{\rho(t)}^2 + &\int_0^t \diff \norm{ \nabla \rho}^2 \ds \leq \delta(\theta)^2 + \expconst \int_0^t \norm{\rho}^2 \ds + C_2 \sup_{s \in [0,t]} \norm{\rho(s)}^\gamma  \int_0^t \norm{\nabla \rho}^2 \dt  	\notag \\
		&\leq \delta(\theta)^2 + \expconst \int_0^t \norm{\rho}^2 \ds + C_2\left( \sup_{s \in [0,t]}\norm{\rho(s)}^2 + \int_0^t \norm{\nabla \rho}^2 \ds \right)^{1 + \frac{\gamma}{2}}.
\end{align}

To bound the final terms on the right hand side using $\delta$, suppose that the maximum size $h_{\max}$ of the mesh used to partition the domain $\domain$ is small enough that for $h < h_{\max}$, the reconstruction error $\theta$ satisfies
\begin{equation*}
	\delta(\theta) \leq C_2^{-\gamma}\left( 4 e^{\expconst T} \right) ^{-\frac{2 + \gamma}{2\gamma}},
\end{equation*}
implying that
\begin{align*}
    C_2 \left( 4 \delta(\theta)^2 e^{\expconst T} \right)^{1 + \frac{\gamma}{2}} \leq \delta(\theta)^2.
\end{align*}
Consider the set
\begin{equation*}
	I = \left\{ \tau \in [0,T] \, : \, \sup_{s \in [0, \tau]} \norm{\rho(s)}^2 + \int_0^\tau \diff \norm{ \nabla \rho}^2 \ds  \leq 4 \delta(\theta)^2 e^{\expconst T} \right\}.
\end{equation*}
Upon observing that, by construction, we have $\rho(0) = 0$, this set is clearly not empty since it contains $\tau = 0$.
Moreover, the continuity of $\rho$ in time implies that $I$ must be closed, and thus the maximum of the set is well defined.
Thus denoting $\tau^* = \max I$, we suppose that $\tau^* < T$.
Then, for $t \leq \tau^*$:
\begin{equation*}
	C_2 \left( \sup_{s \in [0, t]} \norm{\rho(s)}^2 + \int_0^t \diff \norm{ \nabla \rho}^2 \ds \right)^{1 + \frac{\gamma}{2}} \leq C_2 \left( 4 \delta(\theta)^2 e^{\expconst T} \right)^{1 + \frac{\gamma}{2}} \leq \delta(\theta)^2.
\end{equation*}
Substituting this into~\eqref{eq:beforeHiding} we find that, for $t \leq \tau^*$
\begin{equation*}
	\norm{\rho(t)}^2 + \int_0^t \diff \norm{ \nabla \rho}^2 \ds \leq 2\delta(\theta)^2 + \expconst \int_0^t \norm{\rho}^2 \ds,
\end{equation*}
from which Gronwall's inequality implies
\begin{equation*}
	\norm{\rho(t)}^2 + \int_0^t \diff \norm{ \nabla \rho}^2 \ds \leq 2\delta(\theta)^2 e^{\expconst T}.
\end{equation*}
Since this is true for all $t$ (and the integral on the left is non-decreasing), it follows that
\begin{equation*}
	\sup_{s \in [0,t]} \norm{\rho(s)}^2 + \int_0^t \diff \norm{ \nabla \rho}^2 \ds \leq 2\delta(\theta)^2 e^{\expconst T},
\end{equation*}
which contradicts the assumption that $\tau^* < T$ due to the continuity of $\rho$ in time.

Consequently, we have
\begin{equation*}
	\norm{\rho(t)}^2 + \int_0^t \diff \norm{ \nabla \rho}^2 \ds \leq 2 e^{\expconst s} \int_0^t \norm{\theta_t}^2 + 4M \Cf \norm{\theta}^2 + C_1 \norm{\theta}^\gamma \norm{\nabla \theta}^2 \ds,
\end{equation*}
for any $s \in [0,T]$.
Applying the triangle inequality and invoking the bound~\eqref{eq:associatedEllipticBound} for $\theta$ then produces the required bounds.
\end{proof}

We note that the above error bounds are computable. Indeed, assuming the existence of the constant $M$ bounding the $L^\infty$-norm of the finite element solution $\uh$ is not unreasonable since we can assume to have computed it, and thus have access to its maximal and minimal values. Hence both bounds of Theorem~\ref{aposteriori}  are computable since they depend only on the discrete solution and problem data.

The error \emph{indicator} which we use to mark mesh elements for either refinement or coarsening is derived from the a posteriori error bound of Theorem~\ref{aposteriori} and is of the form
\begin{align*}
	\epsilon E_1(\u^{\timestep})^2 + E_0\Big(\frac{\u^{\timestep} - \u^{\timestep-1}}{\tau}\Big)^2,
\end{align*}
which may naturally be broken into contributions from each element by observing the structure of $E_\reg$ in~\eqref{eq:ellipticEstimator}. We remark that a posteriori bounds for the time discretisation by the Crank-Nicolson method are also available \cite{AMN06,BKM12}. Given the nature of the simulations, however, whereby the length-scales present do not change over time, (but only in position,) the incorporation of a full-space time a posteriori analysis in the spirit of \cite{BKM12} was not deemed necessary in this case. Crucially, however, the modified Crank-Nicolson method of \cite{BKM12} was used in the present context of temporal mesh-modification.

Figure~\ref{fig:meshes} shows some examples of the computational meshes used to obtain the results of Section~\ref{sec:results}, reporting the number of elements saved compared to an equivalent uniform mesh in each case (which may be used as a rough estimate of the computational effort required to compute the solution).
What this clearly demonstrates is the effectiveness of the resulting adaptive scheme, since the number of elements required is reduced by over 50\% in each case and typically dramatically more.
Moreover, examining the areas in which the algorithm has opted to refine or coarsen the mesh indicates that computational effort (in the form of high mesh resolution) is tightly targeted at the areas where it is required, demonstrating that this indicator works well in driving the adaptive algorithm.

\begin{figure}%
	\hfil%
	\subcaptionbox{Adapted mesh for the solution shown in Fig.~\ref{fig:circlepatches_240} containing 18,418 elements (73\% saving over equivalent uniform mesh)}[0.45\textwidth]{%
        \includegraphics[width=0.45\textwidth]{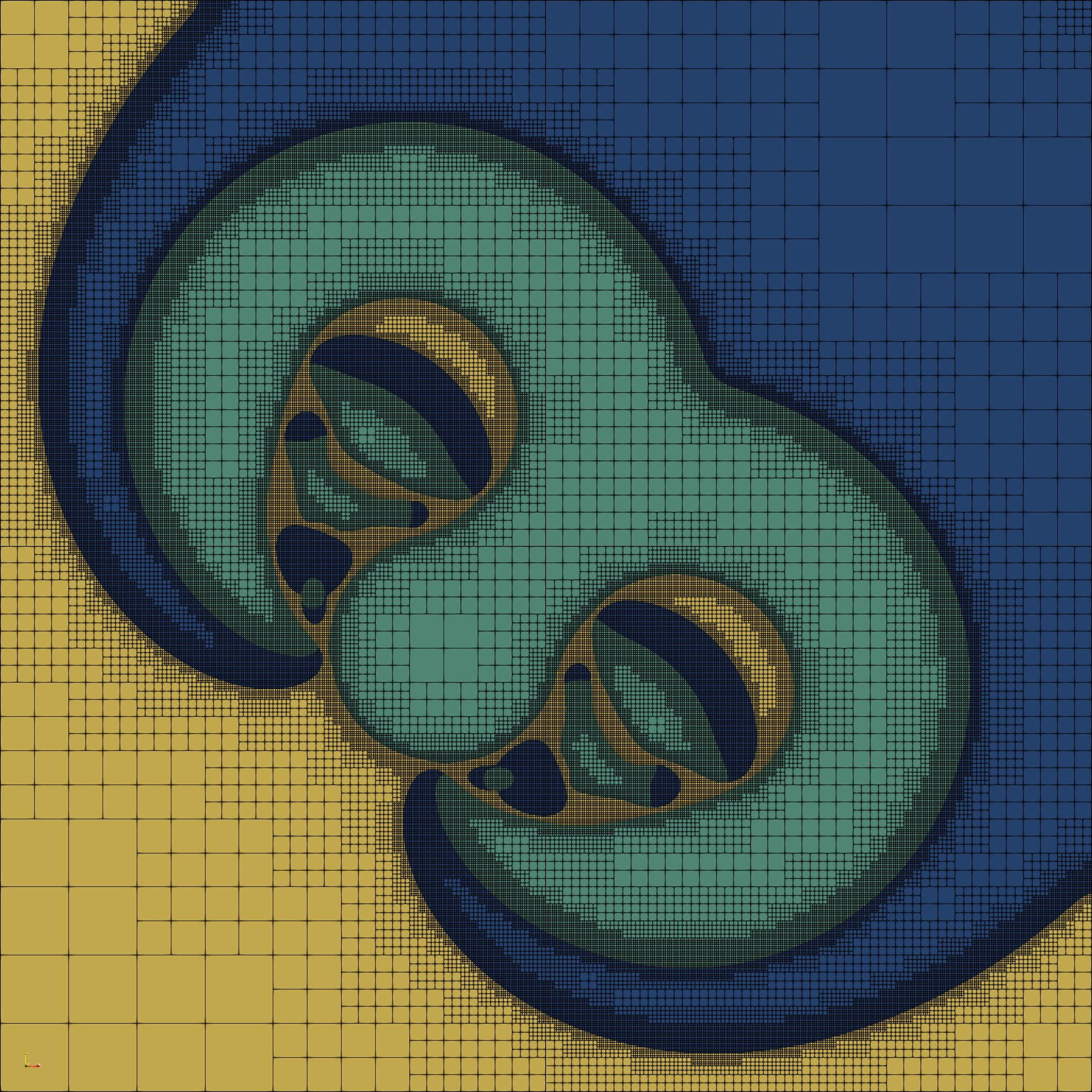}
    }\hfil%
    \subcaptionbox{Adapted mesh for the solution shown in Fig.~\ref{fig:e2=1_e3=1}, containing 8,935 elements (83\% saving over equivalent uniform mesh)}[0.45\textwidth]{%
        \includegraphics[width=0.45\textwidth]{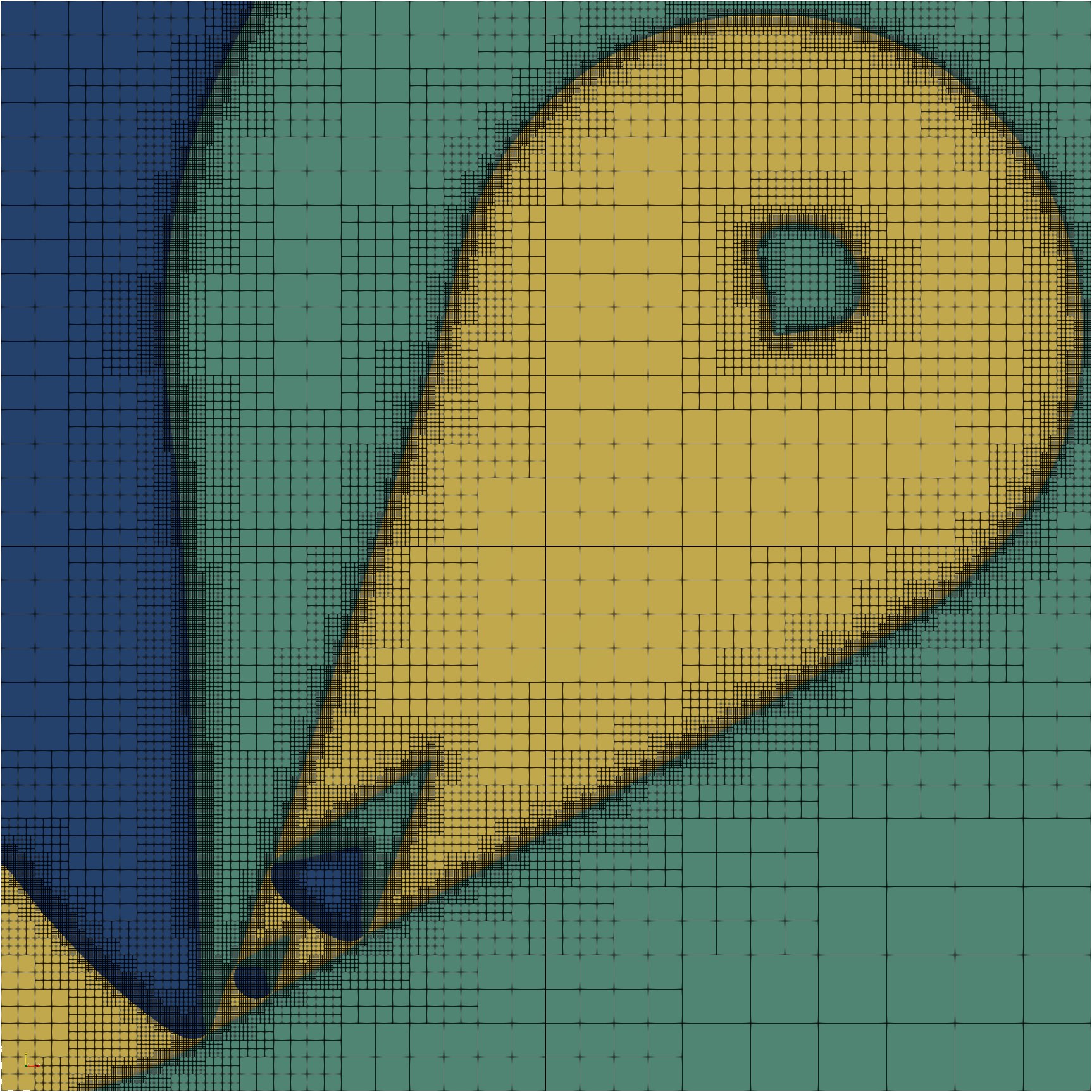}
    }\hfil
    \\[1.0em]
    \hfil%
    \subcaptionbox{Adapted mesh for the solution shown in Fig.~\ref{fig:glider_290}, containing 5,953 elements (91\% saving over equivalent uniform mesh)}[0.45\textwidth]{%
        \includegraphics[width=0.45\textwidth]{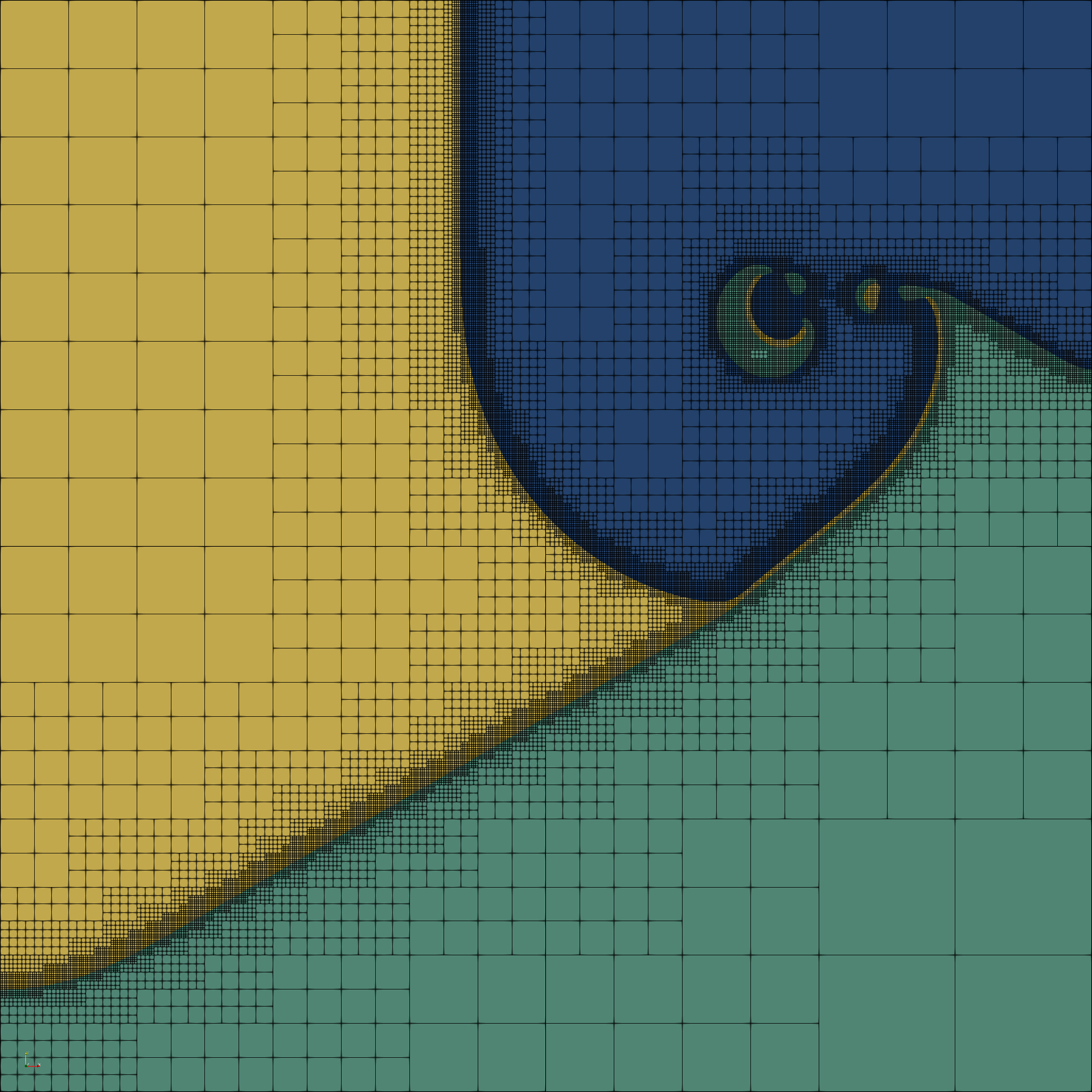}
    }\hfil%
    \subcaptionbox{Adapted mesh for the solution shown in Fig.~\ref{fig:cycliccompetition:3D}, containing 1,015,660 elements (51\% saving over equivalent uniform mesh)\label{fig:cycliccompetition:3Dmesh}}[0.45\textwidth]{%
        \includegraphics[width=0.45\textwidth]{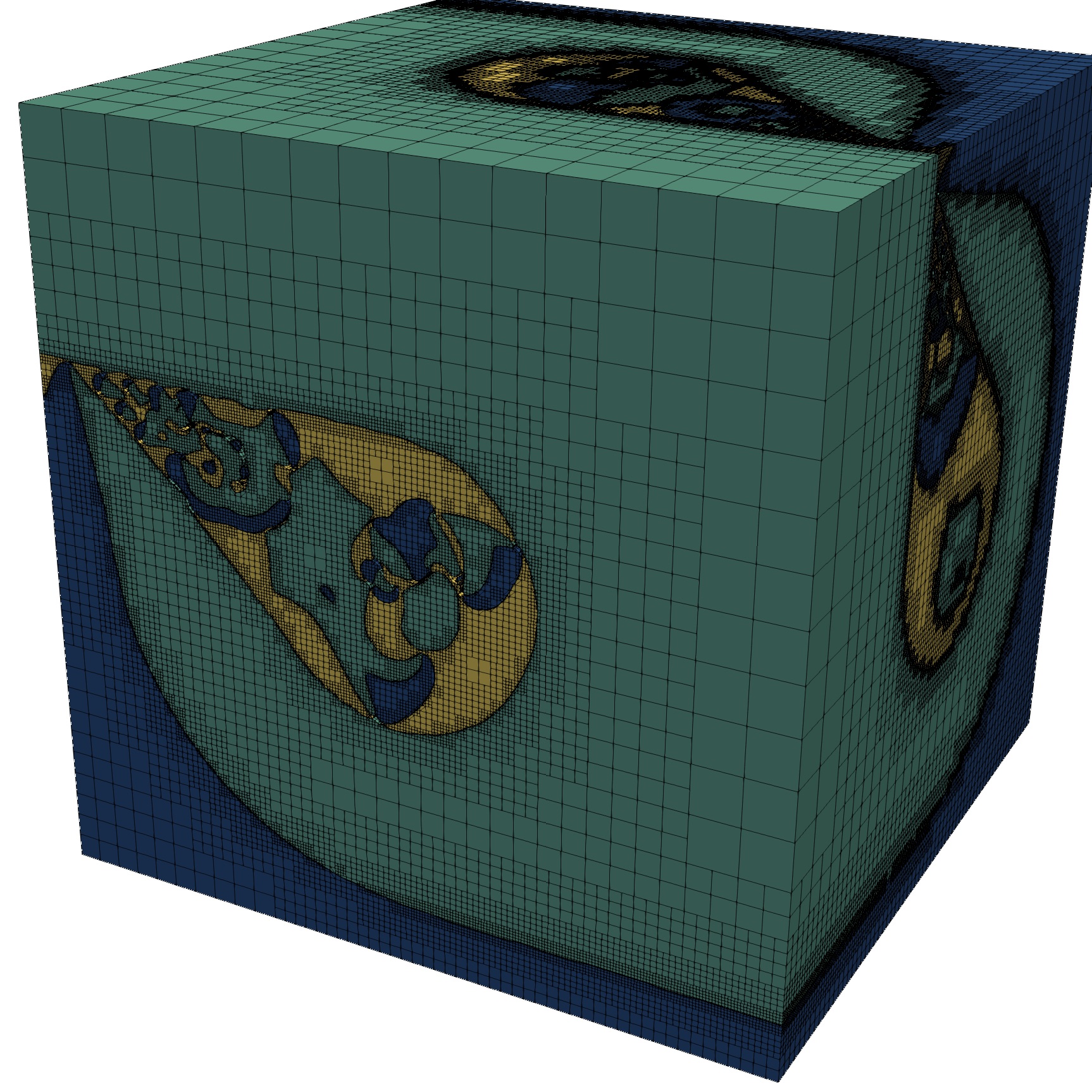}
    }\hfill
    \caption{Some examples of meshes produced by the adaptive algorithm, demonstrating the reduction in the number of elements required compared to a uniform mesh with the same resolution around the layers.}
    \label{fig:meshes}
\end{figure}

\section{Discussion and conclusions}
\label{sec:conclusion}
Although pattern formation and travelling waves in reaction-diffusion models are widely researched, it seems that --- surprisingly --- we still lack a clear understanding of all types of patterns in some well-known systems, such as the Lotka-Volterra three-species competition model~\cite{Petrovskii2001, mimura2015, MuN11, AM12, Ch12, C15}. In this paper, we study such a model with a cyclic competition interaction structure and reveal several novel patterns of population waves in this system which have not been reported before, but may have important applications in biology, chemistry or game theoretical models. To achieve our goal, we utilised a novel, adaptive numerical method driven by an a posteriori error estimate which allowed us to efficiently run simulations in both two and three spatial dimensions.

For a cyclic competition system without space, it is well known that the eventual result is a single exclusive species, the densities of the other species being droped below an extinction threshold  ~\cite{ML75,AM12}.
Adding a spatial dependence, however, allows for long term species coexistence.
A major question is therefore to understand the dynamical patterns which allow the domain in which the species coexist to spread into areas occupied by a single species.
The richness we observe in the system, in terms of its many varied dynamics regimes, stems from the fact that we allow the three species to have different mobilities.
Note that this is a natural case which is observed, for example, in the cyclic system of coral reef invertebrates~\cite{Buss79}, and is modelled here by allowing different diffusion coefficients for each species.

In previous studies of model~\eqref{eq:scaledsystem}, it was reported that the area of mutual coexistence can spread through space via spiral patterns, interacting patches, or different types of travelling wedges~\cite{mimura2015, AM12, C15}. Here we demonstrate new spatially regular patterns of spread. Interestingly, in their work Contento and co-authors ~\cite{C15} hypothesised the existence of droplet-like structures in a spreading wedge which is close to that shown in Fig.~\ref{fig:cycliccompetition:icecream}, although they did not find a practical realisation of such a pattern and assumed that it would be unstable~\cite{C15}. Here we found a stable pattern consisting of droplets in a spreading wedge. It is worth observing, however, that in our case the mechanism by which droplets form is slightly different. For example, in our case the spread is based on pairwise interactions between species and, unlike in the cited work, the corresponding 1D case does not allow the spread of a travelling pulse involving all three species. Moreover, the authors of~\cite{C15} hypothesised that their pattern would exist under \emph{conditional} cyclic competition, while the droplet-shaped structures in Fig.~\ref{fig:cycliccompetition:icecream} are found in \emph{classical} (i.e. unconditional) cyclic competition. It therefore still remains to be determined whether the patterns predicted by Contento and co-authors are actually possible in the case of conditional cyclic competition involving local bistability.

The pattern of spread shown in Fig.~\ref{fig:stripes} is of particular interest, not only because of the apparent regularity in the direction of movement and irregularity in the opposite direction. A novel feature of the pattern seems to be the coexistence of two different waves moving towards the left hand boundary: one wave is a wave of regularity composed of almost parallel bands in the middle and two wedge-shaped waves of chaos on each side of the bands. Our simulations show a long term coexistence of both types of waves. Using this pattern one can describe a complex spread of species involving regular and irregular population patches.

These newly demonstrated patterns of travelling waves with spatially regular structure can be interpreted in terms of the definition of \emph{convective stability} introduced by~\cite{sherratt2014mathematical}. Indeed, the developing regular travelling structures are convectively stable since they emerge as a result of complex spatio-temporal interactions. However, they are not globally stable, as shown in Fig.~\ref{fig:stripes}: depending on the initial conditions, both the waves of regularity and the waves of chaos can be simultaneously realised in the same system.

The spatially regular geometric shapes found by this study to exist in the wake of spreading waves may have applications in the life sciences. It is well known that the distribution of vegetation in semiarid or other regions show regular band-shaped patterns which slowly move over time~\cite{montana1992, lefever97, rietkerk2008}. The common point of view on the origin of these vegetation patterns is the interaction between the soil and plants controlled by the level of moisture via various mechanisms such as Turing pattern formation or periodic travelling waves~\cite{klausmeier99, Dagbovie14}. Here we suggest an alternative mechanism for the formation such bands, due to the interaction of competitive plant species which, for instance, does not require the existence of a steady gradient in the system. 

The pattern containing regular droplet-shaped structures shown in Fig.~\ref{fig:cycliccompetition:icecream} and Fig.~\ref{fig:parameterSearch} can potentially be realised on growing domains~\cite{CGM02, page2005} such as in the pigmentation and relief-like patterns found on mollusc shells, which remain a long standing question~\cite{meinhardt1992, Fowler92}. Previously, it was suggested that regular patterns in mollusc shells are the result of inhibitor-activator type interactions via a Turing mechanism. Here we show that similar patterns can be produced by a cyclic competition type of interaction via \emph{a non-Turing mechanism}. Finally, the transient glider-type patterns shown in Fig.~\ref{fig:cycliccompetition:glider} in the case of conditional cyclic competition provide an example of a new mechanism of patchy spread of invasive species. This new pattern can be used to improve our understanding and modelling of biological invasions since empirical observations often report that the spread of a species into the habitat occupied by another species occurs via the propagation of irregular patches~\cite{shigesada1997}. This also supports the recent ecological theory of \emph{invasional meltdown}, when an invasive species facilitates the invasion of some other invasive species~\cite{o2003}. Note that unlike the original concept of invasional meltdown, suggesting mutual facilitation of invasion of species via mutualistic interactions, in our case we consider the case of antagonistic competitors~\cite{simberloff1999}.

Our results have also allowed us to improve our understanding of the role of dimensionality on the pattern formation in the considered type of models.
This can be seen by comparing Fig.~\ref{fig:cycliccompetition:icecream} and Fig.~\ref{fig:cycliccompetition:glider} alongside the corresponding 1D simulations (not shown here for the sake of simplicity).
In one spatial dimension, a wave of mutual spread of three species is impossible for the given parameters: only pairwise switch waves are observed. With two spatial dimensions, the droplet-shape pattern can emerge even through it is simply the result of the pairwise interaction of plane waves, as shown in Fig.~\ref{fig:schematic} (the round-shaped interface can be formally considered as a plain wave in polar coordinates).
Thus, increasing from one to two spatial dimensions allows for species coexistence through a structure which was previously impossible.
Interestingly, adding a third dimension continues to allow the persistence of the droplet-shape structure, although we argue here that the pattern remains primarily two-dimensional since the three dimensional droplets are observed to have a prismatic structure, and can still be described via pairwise interactions of locally plain prismatic waves.

Bearing in mind the large domains and long time scales required for the full solution dynamics to evolve, it is clear that the computational cost of these simulations would be prohibitively expensive using a uniform mesh, an issue which is amplified in three spatial dimensions.
Instead, the adaptive numerical method described in Section~\ref{sec:numericalscheme}, based on the novel computable error indicator derived in Section~\ref{sec:apost}, allows us to obtain accurate simulations using just a fraction of the computational effort of an equivalent non-adaptive scheme, as demonstrated by the adapted computational meshes shown in Fig.~\ref{fig:meshes}.
The savings this method provides means that highly accurate simulations of this model are within reach of researchers without needing access to high performance computing facilities.
Moreover, since the error analysis of Section~\ref{sec:apost} is applicable to a much wider class of semilinear reaction-diffusion problems, the adaptive method which we describe can be easily applied by researchers wishing to study other phenomena.

We should point out that our numerical investigation of the model cyclic Lotka-Voltera system is by no means exhaustive. 
We do not claim that combined with the previous studies of the system~\cite{mimura2015, MuN11, AM12, Ch12, C15} we have now completed a full classification of possible patterns. 
Further research will be needed specially to further explore the case of conditional cyclic competition. 
Another interesting direction is to further explore the influence of the number of spatial dimensions on the species persistence. 
In other words, it is interesting to verify whether or not adding a third dimension will enhance the coexistence of all species and which possible patterns of mutual coexistence can occur. 
This is a biologically relevant question which is important for understanding, for example, the coexistence of competing bacterial strains or microalgae in laboratory and natural conditions.

\section*{Acknowledgements}
EHG acknowledges financial support by The Leverhulme Trust (grant no.~RPG-2015-306). 
OJS acknowledges financial support by the EPSRC (grant no.~EP/P000835/1).
This research used the ALICE High Performance Computing Facility at the University of Leicester. 
The authors thank Ruslan Davidchack (University of Leicester) for his encouragement and support of this project.

\bibliographystyle{siam}
\bibliography{cyclic-competition-arxiv.bbl}

\begin{thebibliography}{10}

\bibitem{AF03}
{\sc R.~A. Adams and J.~J.~F. Fournier}, {\em Sobolev spaces}, vol.~140 of Pure
  and Applied Mathematics (Amsterdam), Elsevier/Academic Press, Amsterdam,
  second~ed., 2003.

\bibitem{AM12}
{\sc M.~W. Adamson and A.~Y. Morozov}, {\em Revising the role of species
  mobility in maintaining biodiversity in communities with cyclic competition},
  Bull. Math. Biol., 74 (2012), pp.~2004--2031.

\bibitem{AO00}
{\sc M.~Ainsworth and J.~T. Oden}, {\em A posteriori error estimation in finite
  element analysis}, Pure and Applied Mathematics (New York),
  Wiley-Interscience [John Wiley \& Sons], New York, 2000.

\bibitem{AMN06}
{\sc G.~Akrivis, C.~Makridakis, and R.~H. Nochetto}, {\em A posteriori error
  estimates for the {C}rank-{N}icolson method for parabolic equations}, Math.
  Comp., 75 (2006), pp.~511--531.

\bibitem{ADD10}
{\sc E.~O. Alzahrani, F.~A. Davidson, and N.~Dodds}, {\em Travelling waves in
  near-degenerate bistable competition models}, Math. Model. Nat. Phenom., 5
  (2010), pp.~13--35.

\bibitem{A15}
{\sc O.~Aydogmus}, {\em Patterns and transitions to instability in an
  intraspecific competition model with nonlocal diffusion and interaction},
  Math. Model. Nat. Phenom., 10 (2015), pp.~17--29.

\bibitem{BHK07}
{\sc W.~Bangerth, R.~Hartmann, and G.~Kanschat}, {\em deal.{II}---a
  general-purpose object-oriented finite element library}, ACM Trans. Math.
  Software, 33 (2007), pp.~Art. 24, 27.

\bibitem{BKM12}
{\sc E.~B\"ansch, F.~Karakatsani, and C.~Makridakis}, {\em A posteriori error
  control for fully discrete {C}rank-{N}icolson schemes}, SIAM J. Numer. Anal.,
  50 (2012), pp.~2845--2872.

\bibitem{Buss79}
{\sc L.~W. Buss and J.~B.~C. Jackson}, {\em Competitive networks: Nontransitive
  competitive relationships in cryptic coral reef environments}, Am. Nat., 113
  (1979), pp.~223--234.

\bibitem{CGJ13}
{\sc A.~Cangiani, E.~H. Georgoulis, and M.~Jensen}, {\em Discontinuous
  {G}alerkin methods for mass transfer through semipermeable membranes}, SIAM
  J. Numer. Anal., 51 (2013), pp.~2911--2934.

\bibitem{CGKM16}
{\sc A.~Cangiani, E.~H. Georgoulis, I.~Kyza, and S.~Metcalfe}, {\em Adaptivity
  and blow-up detection for nonlinear evolution problems}, SIAM J. Sci.
  Comput., 38 (2016), pp.~A3833--A3856.

\bibitem{Ch12}
{\sc C.-C. Chen, L.-C. Hung, M.~Mimura, and D.~Ueyama}, {\em Exact travelling
  wave solutions of three-species competition-diffusion systems}, Discrete
  Contin. Dyn. Syst. Ser. B, 17 (2012), pp.~2653--2669.

\bibitem{C}
{\sc P.~G. Ciarlet}, {\em The finite element method for elliptic problems},
  vol.~40 of Classics in Applied Mathematics, Society for Industrial and
  Applied Mathematics (SIAM), Philadelphia, PA, 2002.
\newblock Reprint of the 1978 original [North-Holland, Amsterdam; MR0520174 (58
  \#25001)].

\bibitem{C15}
{\sc L.~Contento, M.~Mimura, and M.~Tohma}, {\em Two-dimensional traveling
  waves arising from planar front interaction in a three-species
  competition-diffusion system}, Jpn. J. Ind. Appl. Math., 32 (2015),
  pp.~707--747.

\bibitem{Cooper12}
{\sc S.~B. Cooper and P.~K. Maini}, {\em The mathematics of nature at the alan
  turing centenary}, Interface Focus, 2 (2012), pp.~393--396.

\bibitem{CGM02}
{\sc E.~J. Crampin, E.~A. Gaffney, and P.~K. Maini}, {\em Mode-doubling and
  tripling in reaction-diffusion patterns on growing domains: a piecewise
  linear model}, J. Math. Biol., 44 (2002), pp.~107--128.

\bibitem{Dagbovie14}
{\sc A.~S. Dagbovie and J.~A. Sherratt}, {\em Pattern selection and hysteresis
  in the rietkerk model for banded vegetation in semi-arid environments}, ‎J.
  R. Soc. Interface, 11 (2014).

\bibitem{Fisher37}
{\sc R.~A. FISHER}, {\em The wave of advance of advantageous genes}, Annals of
  Eugenics, 7 (1937), pp.~355--369.

\bibitem{Fowler92}
{\sc D.~R. Fowler, H.~Meinhardt, and P.~Prusinkiewicz}, {\em Modeling
  seashells}, SIGGRAPH Comput. Graph., 26 (1992), pp.~379--387.

\bibitem{HDHS02}
{\sc C.~Hauert, S.~De~Monte, J.~Hofbauer, and K.~Sigmund}, {\em Replicator
  dynamics for optional public good games}, J. Theoret. Biol., 218 (2002),
  pp.~187--194.

\bibitem{Hosono1998}
{\sc Y.~Hosono}, {\em The minimal speed of traveling fronts for a diffusive
  lotka-volterra competition model}, Bulletin of Mathematical Biology, 60
  (1998), p.~435—448.

\bibitem{kirkup04}
{\sc B.~C. Kirkup and M.~A. Riley}, {\em Antibiotic-mediated antagonism leads
  to a bacterial game of rock-paper-scissor in vivo}, Nature, 428 (2004),
  p.~412.

\bibitem{K82}
{\sc K.~Kishimoto}, {\em The diffusive {L}otka-{V}olterra system with three
  species can have a stable nonconstant equilibrium solution}, J. Math. Biol.,
  16 (1982/83), pp.~103--112.

\bibitem{klausmeier99}
{\sc C.~A. Klausmeier}, {\em Regular and irregular patterns in semiarid
  vegetation}, Science, 284 (1999), pp.~1826--1828.

\bibitem{kolmogorov37}
{\sc A.~Kolmogorov, I.~Petrovsky, and N.~Piskunov}, {\em Investigation of the
  equation of diffusion combined with increasing of the substance and its
  application to a biology problem}, Bull. Moscow State Univ. Ser. A: Math.
  Mech, 1 (1937), pp.~1--25.

\bibitem{lefever97}
{\sc R.~Lefever and O.~Lejeune}, {\em On the origin of tiger bush}, Bull. Math.
  Biol., 59 (1997), pp.~263--294.

\bibitem{MN03}
{\sc C.~Makridakis and R.~H. Nochetto}, {\em Elliptic reconstruction and a
  posteriori error estimates for parabolic problems}, SIAM J. Numer. Anal., 41
  (2003), pp.~1585--1594.

\bibitem{ML75}
{\sc R.~M. May and W.~J. Leonard}, {\em Nonlinear aspects of competition
  between three species}, SIAM J. Appl. Math., 29 (1975), pp.~243--253.
\newblock Special issue on mathematics and the social and biological sciences.

\bibitem{meinhardt1992}
{\sc H.~Meinhardt}, {\em Pattern formation in biology: a comparison of models
  and experiments}, Rep. Prog. Phys, 55 (1992), p.~797.

\bibitem{MF86}
{\sc M.~Mimura and P.~C. Fife}, {\em A {$3$}-component system of competition
  and diffusion}, Hiroshima Math. J., 16 (1986), pp.~189--207.

\bibitem{mimura2015}
{\sc M.~Mimura and M.~Tohma}, {\em Dynamic coexistence in a three-species
  competition--diffusion system}, Ecol. Compl., 21 (2015), pp.~215--232.

\bibitem{montana1992}
{\sc C.~Montana}, {\em The colonization of bare areas in two-phase mosaics of
  an arid ecosystem}, J. Ecol.,  (1992), pp.~315--327.

\bibitem{morozov2007}
{\sc A.~Morozov and B.-L. Li}, {\em On the importance of dimensionality of
  space in models of space-mediated population persistence}, Theor. Pop. Biol.,
  71 (2007), pp.~278--289.

\bibitem{MuN11}
{\sc H.~Murakawa and H.~Ninomiya}, {\em Fast reaction limit of a
  three-component reaction-diffusion system}, J. Math. Anal. Appl., 379 (2011),
  pp.~150--170.

\bibitem{o2003}
{\sc D.~J. O'Dowd, P.~T. Green, and P.~S. Lake}, {\em Invasional `meltdown' on
  an oceanic island}, Ecol. Lett., 6 (2003), pp.~812--817.

\bibitem{page2005}
{\sc K.~M. Page, P.~K. Maini, and N.~A. Monk}, {\em Complex pattern formation
  in reaction--diffusion systems with spatially varying parameters}, Phys. D,
  202 (2005), pp.~95--115.

\bibitem{paquin1983}
{\sc C.~E. Paquin and J.~Adams}, {\em Relative fitness can decrease in evolving
  asexual populations of s-cerevisiae},  (1983).

\bibitem{Petrovskii2001}
{\sc S.~Petrovskii, K.~Kawasaki, F.~Takasu, and N.~Shigesada}, {\em Diffusive
  waves, dynamical stabilization and spatio-temporal chaos in a community of
  three competitive species}, Jpn. J. Ind. Appl. Math., 18 (2001),
  pp.~459--481.

\bibitem{petrovskii2002}
{\sc S.~V. Petrovskii, A.~Y. Morozov, and E.~Venturino}, {\em Allee effect
  makes possible patchy invasion in a predator--prey system}, Ecol. Lett., 5
  (2002), pp.~345--352.

\bibitem{rietkerk2008}
{\sc M.~Rietkerk and J.~Van~de Koppel}, {\em Regular pattern formation in real
  ecosystems}, TREE, 23 (2008), pp.~169--175.

\bibitem{satnoianu2000}
{\sc R.~A. Satnoianu, M.~Menzinger, and P.~K. Maini}, {\em Turing instabilities
  in general systems}, J. Math. Biol., 41 (2000), pp.~493--512.

\bibitem{sherratt2014mathematical}
{\sc J.~A. Sherratt, A.~S. Dagbovie, and F.~M. Hilker}, {\em A mathematical
  biologist’s guide to absolute and convective instability}, Bull. Math.
  Biol., 76 (2014), pp.~1--26.

\bibitem{shigesada1997}
{\sc N.~Shigesada and K.~Kawasaki}, {\em Biological invasions: theory and
  practice}, Oxford University Press, UK, 1997.

\bibitem{simberloff1999}
{\sc D.~Simberloff and B.~Von~Holle}, {\em Positive interactions of
  nonindigenous species: invasional meltdown?}, Biol. Invas., 1 (1999),
  pp.~21--32.

\bibitem{sinervo1996}
{\sc B.~Sinervo and C.~M. Lively}, {\em The rock-paper-scissors game and the
  evolution of alternative male strategies}, Nature, 380 (1996), p.~240.

\bibitem{tainaka1993}
{\sc K.-I. Tainaka}, {\em Paradoxical effect in a three-candidate voter model},
  Phys. Lett. A, 176 (1993), pp.~303--306.

\bibitem{tsyganov2014}
{\sc M.~Tsyganov and V.~Biktashev}, {\em Classification of wave regimes in
  excitable systems with linear cross diffusion}, Physical Review E, 90 (2014),
  p.~062912.

\bibitem{turing1952}
{\sc A.~M. Turing}, {\em The chemical basis of morphogenesis}, Phil. Trans. R
  Soc. Lond B, 237 (1952), pp.~37--72.

\bibitem{V13}
{\sc R.~Verf\"urth}, {\em A posteriori error estimation techniques for finite
  element methods}, Numerical Mathematics and Scientific Computation, Oxford
  University Press, Oxford, 2013.

\bibitem{volpert2009}
{\sc V.~Volpert and S.~Petrovskii}, {\em Reaction--diffusion waves in biology},
  Phys. Life Rev., 6 (2009), pp.~267--310.

\bibitem{wilson1995}
{\sc W.~Wilson, E.~McCauley, and A.~De~Roos}, {\em Effect of dimensionality on
  lotka-volterra predator-prey dynamics: Individual based simulation results},
  Bull. Math. Biol., 57 (1995), pp.~507--526.

\bibitem{W87}
{\sc J.~Wloka}, {\em Partial differential equations}, Cambridge University
  Press, Cambridge, 1987.
\newblock Translated from the German by C. B. Thomas and M. J. Thomas.

\end{thebibliography}

\end{document}